%% file: TR.tex
\documentclass[10pt,a4paper,british]{scrartcl}
\usepackage[T1]{fontenc}
\usepackage[latin1]{inputenc}
\usepackage{verbatim}
\usepackage{subfigure}
\usepackage{floatflt}
\usepackage{float}
\usepackage{graphicx}
\usepackage{amsmath}
\usepackage{amssymb}
\makeatletter

\newcommand{\noun}[1]{\textsc{#1}}

\floatstyle{ruled}
\newfloat{algorithm}{tbp}{loa}
\floatname{algorithm}{Algorithm}

\usepackage{a4wide}
\usepackage{fancyhdr}
\RequirePackage{ifpdf}
\ifpdf
  \usepackage[thmmarks,amsmath]{ntheorem}  
   \usepackage[verbose,
                      pdftex,
                      colorlinks=true,          
                      linkcolor=black,          
                      filecolor=black,           
                      citecolor=black,  
                      urlcolor=black,      
                      bookmarksopen=true,
                      raiselinks=true,
                      breaklinks=true,
                      bookmarksopenlevel=2,
                      bookmarksnumbered=true,
                      pdfpagemode=UseOutlines,
                      plainpages=true,
                      hyperindex=true, 
                      pdfstartview=FitH,
                      linktocpage=true]{hyperref}
\hypersetup{
pdftitle={Discriminated Belief Propagation},
pdfsubject={Processing of Binary Information.},
pdfauthor={\textcopyright\ Ulrich Sorger, ulrich.sorger at uni.lu},
pdfcreator={\LaTeX\ (LyX)\ },
pdfkeywords={Iterative Decoding, Coupled Codes, Information Theory, Complexity, Belief Propagation, Typical Decoding, Set Representations, Central Limit Theorem, Equalisation, Estimation, Trellis Algorithms},
            }
\else
   \usepackage[amsmath]{ntheorem}
\fi
\usepackage{lettrine}
\usepackage{times} 
\usepackage{bm}
\usepackage{graphics}
\renewcommand{\mathbf}[1]{\bm{#1}}
\renewcommand{\[}{$$}
\renewcommand{\]}{$$}
\usepackage{babel}
\makeatother
\begin{document}
\title{Discriminated Belief Propagation}
\author{Uli Sorger
\ifpdf 
\thanks{Technical Report  \href{http://wiki.uni.lu/csc/Discriminated+Belief+Propagation.html}{TR-CSC-07-01} of the D\_Max Project funded by the  \href{http://www.uni.lu}{University of Luxembourg}.}
\else
\thanks{Technical Report TR-CSC-07-01 of the D\_Max Project funded by the Universtity of Luxembourg.}
\fi
}
\maketitle
\begin{abstract}
Near optimal decoding of good error control codes is generally
a difficult task. However, for a certain type of (sufficiently) good
codes an efficient decoding algorithm with near optimal performance
exists. These codes are defined via a combination of constituent codes
with low complexity trellis representations. Their decoding algorithm
is an instance of (loopy) belief propagation and is based on an iterative
transfer of constituent beliefs. The beliefs are thereby given by
the symbol probabilities computed in the constituent trellises. Even
though weak constituent codes are employed close to optimal performance
is obtained, i.e., the encoder/decoder pair (almost) achieves the
information theoretic capacity. However, (loopy) belief propagation
only performs well for a rather specific set of codes, which limits
its applicability.\\
In this paper a generalisation of iterative decoding is presented.
It is proposed to transfer more values than just the constituent beliefs.
This is achieved by the transfer of beliefs obtained by independently
investigating parts of the code space. This leads to the concept of
discriminators, which are used to improve the decoder resolution within
certain areas and defines discriminated symbol beliefs. It is shown
that these beliefs approximate the overall symbol probabilities. This
leads to an iteration rule that (below channel capacity) typically
only admits the solution of the overall decoding problem. Via a \noun{Gauss}
approximation a low complexity version of this algorithm is derived.
Moreover, the approach may then be applied to a wide range of channel
maps without significant complexity increase.
\end{abstract}


\input{TR-intro.tex}

\input{TR-hard.tex}

\input{TR-gauss.tex}

\input{TR-channel.tex}

\input{TR-summary.tex}

\clearpage
\input{TR-trellis.tex}

\bibliographystyle{plain}
\bibliography{ref2,referenzen}
\end{document}

%% file: TR-intro.tex
\pagestyle{plain} \lettrine{D}{ecoding} error control codes is the
inversion of the encoding map in the presence of errors. An optimal
decoder finds the codeword with the least number of errors. However,
optimal decoding is generally computationally infeasible due to the
intrinsic non linearity of the inversion operation. Up to now only
simple codes can be optimally decoded, e.g., by a simple trellis representation.
These codes generally exhibit poor performance or rate~\cite{DBLP:journals/tit/KschischangS95}. 

On the other hand, good codes can be constructed by a combination
of simple constituent codes (see e.g.,~\cite[pp.567ff]{MacWillian_Sloan}).
This construction is interesting as then a trellis based inversion
may perform almost optimally: \noun{Berrou} et al.~\cite{Berrou}
showed that iterative turbo decoding leads to near capacity performance.
The same holds true for iterative decoding of Low Density Parity Check
(LDPC) codes~\cite{Gallager62}. Both decoders are conceptually similar
and based on the (loopy) propagation of beliefs~\cite{McElice98}
computed in the constituent trellises. However, (loopy) belief propagation
is often limited to idealistic situations. E.g., turbo decoding generally
performs poorly for multiple constituent codes, complex channels,
good constituent codes, and/or relatively short overall code lengths.

In this paper a concept called discrimination is used to generalise
iterative decoding by (loopy) belief propagation. The generalisation
is based on an uncertainty or distance discriminated investigation
of the code space. The overall results of the approach are linked
to basic principles in information theory such as typical sets and
channel capacity~\cite{Shannon_Inf,Massey_Inf,Mackay}. 

\textbf{Overview:} The paper is organised as follows: First the combination
of codes together with the decoding problem and its relation to belief
propagation are reviewed. Then the concept of discriminators together
with the notion of a common belief is introduced. In the second section
local discriminators are discussed. By a local discriminator a controllable
amount of parameters (or generalised beliefs) are transferred. It
is shown that this leads to a practically computable common belief
that may be used in an iteration. Moreover, a fixed point of the obtained
iteration is typically the optimal decoding decision. Section~3 finally
considers a low complexity approximation and the application to more
complex channel maps.

%% file: TR-hard.tex
\section{Code Coupling\label{sec:Coupled-Codes}}

To review the combination of constituent codes we here consider only
binary linear codes $\mathbf{C}$ given by the encoding map\[
\mathcal{C}:\mathbf{x}=(x_{1},\ldots,x_{k})\mapsto\mathbf{c}=(c_{1},\ldots,c_{n})=\mathbf{x}\mathbf{G}\mbox{ mod }2\]
with $\mathbf{G}$ the $(k\times n)$ \emph{generator} matrix with
$x_{i}$, $c_{i}$, and $G_{i,j}\in\mathbb{Z}_{2}=\{0,1\}$. 

The map defines for $\mbox{rank}(\mathbf{G})=k$ the \emph{event}
set $\mathbb{E}(\mathbf{C})$ of $2^{k}$ code words $\mathbf{c}$.
The \emph{rate} of the code is given by $R=k/n$ and it is for an
error correcting code smaller than One.

The event set $\mathbb{E}(\mathbf{C})$ is by linear algebra equivalently
defined by a $((n-k)\times n)$ \emph{parity} matrix $\mathbf{H}$
with $\mathbf{H}\mathbf{G}^{T}=\mathbf{0}\mbox{ mod }2$ and thus\[
\mathbb{E}(\mathbf{C})=\{\mathbf{c}:\mathbf{H}\mathbf{c}^{T}=\mathbf{0}\mbox{ mod }2\}.\]
Note that the modulo operation is in the sequel not explicitly stated. 

$\mathbb{E}(\mathbf{C})$ is a subset of the set $\mathbb{S}$ of
all $2^{n}$ binary vectors of length $n$. The restriction to a subset
is interesting as this leads to the possibility to correct corrupted
words. However, the correction is a difficult operation and can usually
only be practically performed for simple or short codes. 

On the other hand long codes can be constructed by the use of such
simple \emph{constituent} codes. Such constructions are reviewed in
this section. \emph{ }

\begin{Def}
(Direct Coupling) The two constituent linear systematic coding maps
\[
\mathcal{C}^{(l)}:\mathbf{x}\mapsto\mathbf{c}^{(l)}=\mathbf{x}\cdot\mathbf{G}^{(l)}=\mathbf{x}\cdot[\mathbf{I}\,\mathbf{P}^{(l)}]\,\textrm{ with }l=1,2\]
 and a direct coupling gives the overall code $\mathbb{E}(\mathbf{C}^{(a)})$
with $\mathbf{c}^{(a)}=\mathbf{x}\cdot[\mathbf{I}\,\mathbf{P}^{(1)}\,\mathbf{P}^{(2)}]\textrm{.}$
\end{Def}
\begin{minipage}[c][1\totalheight]{0.6\columnwidth}%
\begin{exa}
The constituent codes used for turbo decoding~\cite{Berrou} are
two systematic convolutional codes~\cite{Joh_Faltungscodes} with
low \emph{trellis} decoding complexity (See Appendix~\ref{sub:Trellis-Based-Algorithms}).
The overall code is obtained by a direct coupling as depicted in the
figure to the right. The encoding of the non--systematic part $\mathbf{P}^{(l)}$
can be done by a recursive encoder. The $\Pi$ describes a permutation
of the input vector $\mathbf{x}$, which significantly improves the
overall code properties but does not affect the complexity of the
constituent decoders. If the two codes have rate $1/2$ then the overall
code will have rate~$1/3$. 
\end{exa}
\end{minipage}%
\hfill{} %
\begin{minipage}[c][1\totalheight]{0.39\columnwidth}%
\begin{flushright}
\includegraphics{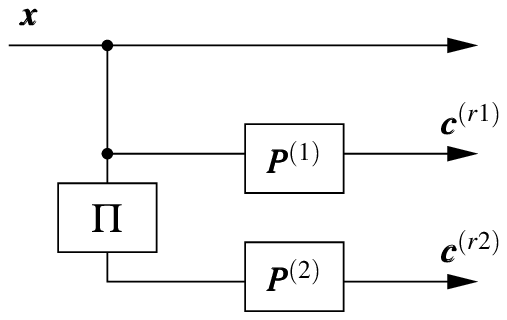}
\par\end{flushright}%
\end{minipage}%

Another possibility is to \emph{concatenate} two constituent codes
as defined below.

\begin{Def}
(Concatenated Codes) By \[
\mathbf{c}^{(1)}=\mathbf{x}\mathbf{G}^{(1)}\textrm{ and }\mathbf{c}^{(a)}=\mathbf{c}^{(1)}\mathbf{G}^{(2)}=\mathbf{x}\mathbf{G}^{(1)}\mathbf{G}^{(2)}\]
(provided matching dimensions, i.e. a $(k\times n^{(1)})$ generator
matrix $\mathbf{G}^{(1)}$ and a $(n^{(1)}\times n)$ generator matrix
$\mathbf{G}^{(2)}$) a concatenated code is given.
\end{Def}
\begin{anm}
(\emph{Generalised Concatenation}) A concatenation can be used to
construct codes with defined properties as usually a large minimum
\noun{Hamming} distance. Note that \emph{generalised} concatenated~\cite{BloZy74,Dumer}
codes exhibit the same basic concatenation map. There distance properties
are investigated under an additional \emph{partitioning} of code $\mathbf{G}^{(2)}$. 
\end{anm}
Another possibility to couple codes is given in the following definition.
This method will show to be very general, albeit rather non intuitive
as the description is based on parity check matrices $\mathbf{H}$.

\begin{Def}
\label{Def:The-dual-coupling} (Dual Coupling) The overall code \emph{\[
\mathbb{C}^{(a)}:=\mathbb{E}(\mathbf{C}^{(a)})=\left\{ \mathbf{c}:\left[\begin{array}{c}
\mathbf{H}^{(1)}\\
\mathbf{H}^{(2)}\end{array}\right]\mathbf{c}^{T}=\mathbf{0}\right\} \]
}is obtained by a dual coupling of the constituent codes $\mathbb{C}^{(l)}:=\mathbb{E}(\mathbf{C}^{(l)})=\{\mathbf{c}:\mathbf{H}^{(l)}\mathbf{c}^{T}=\mathbf{0}\}$
for $l=1,2.$ 
\end{Def}
By a dual coupling the obtained code space is obtained by the \emph{intersection}
$\mathbb{C}^{(a)}=\mathbb{C}^{(1)}\cap\mathbb{C}^{(2)}$ of the constituent
code spaces. 

\begin{exa}
A dually coupled code construction similar to turbo codes is to use
two mutually permuted rate $2/3$ convolutional codes. The intersection
of these two codes gives a code with rate at least $1/3$. To obtain
a larger rate one may employ \emph{puncturing} (not transmitting certain
symbols). However, the encoding of the overall code is not as simple
as for direct coupling codes. A straightforward way is to just use
the generator matrix representation of the overall code. 
\end{exa}
\begin{anm}
(\emph{LDPC} \emph{Codes}) LDPC codes are originally defined by a
single parity check matrix with low weight rows (and columns). An
equivalent representation is via a graph of check nodes (one for each
column) and variables nodes (one for each row). This leads to a third
equivalent representation with two dually coupled constituent codes
and a subsequent puncturing~\cite{Ingmar}. The first constituent
code is thereby given by a juxtaposition of repetition codes that
represent the variable nodes (all node inputs need to be equal). The
second one is defined by single parity check codes representing the
check nodes. The puncturing at the end has to be done such that only
one symbol per repetition code (code column) remains.
\end{anm}
\begin{thm}
\label{thm:Both-direct-coupling-are-dual}Both direct coupling and
concatenated codes are special cases of dual coupling codes.
\end{thm}
\begin{proof}
The direct coupling code is equivalently described in the parity check
form $\mathbf{H}^{(a)}\mathbf{G}^{(a)T}=\mathbf{0}$  by the parity
check matrix \[
\mathbf{H}^{(a)}=\left[\begin{array}{ccc}
\mathbf{H}^{(s1)} & \mathbf{H}^{(r1)} & \mathbf{0}\\
\mathbf{H}^{(s2)} & \mathbf{0} & \mathbf{H}^{(r2)}\end{array}\right]\textrm{ where }\mathbf{H}^{(l)}=[\mathbf{H}^{(sl)}\,\mathbf{H}^{(rl)}]\mbox{ for }l=1,2\]
is the parity check matrix of $\mathbf{G}^{(l)}$ consisting of systematic
part $\mathbf{H}^{(sl)}$ and redundant part $\mathbf{H}^{(rl)}$.
This is obviously a dual coupling. For a concatenated code with systematic
code $\mathbf{G}^{(2)}=\mathbf{[I}\,\mathbf{P}^{(2)}]$ the equivalent
description by a parity check matrix is \[
\mathbf{H}^{(a)}=\left[\begin{array}{cc}
\mathbf{H}^{(1)} & \mathbf{0}\\
\mathbf{H}^{(s2)} & \mathbf{H}^{(r2)}\end{array}\right]\textrm{ with }\mathbf{H}^{(1)}\textrm{ and }\mathbf{H}^{(2)}=[\mathbf{H}^{(s2)}\,\mathbf{H}^{(r2)}]\]
 the parity check matrix of $\mathbf{G}^{(1)}$ respectively $\mathbf{G}^{(2)}$.
For non-systematic concatenated codes a virtual systematic extension
(punctured prior to the transmission) is needed~\cite{Ingmar}. Hence,
a representation by a dual coupling is again possible. 
\end{proof}
It is thus sufficient to consider only dual code couplings. The {}``dual''
is therefore mostly omitted in the sequel. 

\begin{anm}
\label{anm:(Multiple-Dual-Codes)}(\emph{Multiple} \emph{Dual} \emph{Codes})
More than two codes can be dually coupled as described above: By\[
\mathbb{C}^{(a)}=\mathbb{C}^{(1)}\cap\mathbb{C}^{(2)}\cap\mathbb{C}^{(3)}\]
a coupling of three codes is given. The overall parity check matrix
is there given by the juxtaposition of the three constituent parity
check matrix. Multiple dual couplings are produced by multiple intersections.
In the sequel mostly dual couplings with two constituent codes are
considered.
\end{anm}

\subsection{Optimal Decoding}

As stated above the main difficulty is not the encoding but the decoding
of a corrupted word. This corruption is usually the result of a transmission
of the code word over a \emph{channel}. 

\begin{anm}
(\emph{Channels}) \label{anm:BSC_AWGN}In the sequel we assume that
the code symbols $C_{i}$ are in $\mathbb{B}=\{-1,+1\}$. This is
achieved by the use of the {}``BPSK''-map \[
\mathcal{B}:x\mapsto y=\begin{cases}
+1 & \textrm{for }x=0\\
-1 & \textrm{for }x=1\end{cases}\]
prior to the transmission. As channel we assume either a Binary Symmetric
Channel (BSC) with channel error probability $p$ and\begin{align*}
P(\mathbf{r}|\mathbf{s}) & =\prod_{i=1}^{n}(1-p)^{\left\langle s_{i}=r_{i}\right\rangle }p^{\left\langle s_{i}\neq r_{i}\right\rangle }\propto\prod_{i=1}^{n}(\frac{1-p}{p})^{s_{i}r_{i}}=\prod_{i=1}^{n}\exp_{2}(s_{i}r_{i}\log_{2}(\frac{1-p}{p}))\\
 & =\mbox{exp}_{2}(\mathrm{K}\sum_{i=1}^{n}s_{i}r_{i})\mbox{ with }\mathrm{K}=\log_{2}(\frac{1-p}{p})\mbox{ and }s_{i},r_{i}\in\mathbb{B}=\{-1,+1\}\end{align*}
with \begin{align*}
\left\langle b\right\rangle  & =\begin{cases}
0 & \mbox{if }b\mbox{ false}\\
1 & \mbox{if }b\mbox{ true}\end{cases}\end{align*}
or a channel with Additive White \noun{Gauss} Noise (AWGN) given by
\[
P(\mathbf{r}|\mathbf{s})\propto\prod_{i=1}^{n}2^{-(r_{i}-s_{i})^{2}}\propto\exp_{2}(\sum_{i=1}^{n}r_{i}s_{i})\mbox{ and }s_{i}\in\mathbb{B}\]
(this actually is the \noun{Gauss} probability \emph{density}) and
the by $2\sigma_{E}^{2}=\log_{2}(e)$ normalised noise variance. The
received elements $r_{i}$ are in the AWGN case real valued, i.e.,
$r_{i}\in\mathbb{R}$. \\
Note that the normalised noise variance is obtained by $r_{i}^{(l)}\leftarrow\mathrm{K}r_{i}^{(l)}$
and an appropriate constant $\mathrm{K}$. Moreover, then both cases
coincide.
\end{anm}
\noindent Overall this gives that decoding is based on ~$\begin{cases}
1)\mbox{ the knowledge of the code space }\mathbb{E}(\mathbf{C}),\\
2)\mbox{ the knowledge of the channel map given by }P(\mathbf{r}|\mathbf{c}),\mbox{ and}\\
3)\mbox{ the received information represented by }\mathbf{r}.\end{cases}$

A decoding can be performed by a \emph{decision} for some word $\hat{\mathbf{c}}$,
which is in the Maximum Likelihood (ML) word decoding case \[
\hat{\mathbf{c}}=\arg\max_{\mathbf{c}\in\mathbb{E}(\mathbf{C})}P(\mathbf{r}|\mathbf{c})\]
or decisions on the code symbols by ML symbol by symbol decoding \[
\bar{c}_{i}=\arg\max_{x\in\mathbb{B}}P_{C_{i}}^{(c)}(x|\mathbf{r})=\arg\max_{x\in\mathbb{B}}\sum_{\mathbf{c}\in\mathbb{E}(\mathbf{C}),\, c_{i}=x}P(\mathbf{r}|\mathbf{c}).\]
Here $P_{C_{i}}^{(c)}(x|\mathbf{r})$ is the probability that  $c_{i}=x$
under the knowledge of the code space $\mathbb{E}(\mathbf{C})$. If
no further prior knowledge about the code map or other additional
information is available then these decisions are obviously optimal,
i.e., the decisions exhibit smallest word respectively Bit error probability. 

\begin{anm}
\label{anm:(Dominating-ML-Word)}(\emph{Dominating} \emph{ML} \emph{Word})
If by $P^{(a)}(\hat{\mathbf{c}}^{(a)}|\mathbf{r})\to1$ a dominating
ML word decision exists then necessarily holds that $\hat{\mathbf{c}}^{(a)}=\bar{\mathbf{c}}^{(a)}$.
The decoding problem is then equivalent to solving either of the ML
decisions. \\
ML word decoding is for the BSC equivalent to find the code word with
the smallest number of errors $c_{i}\neq r_{i}$, respectively the
smallest \noun{Hamming} distance $d_{H}(\mathbf{c},\mathbf{r})$.
For the AWGN channel the word $\mathbf{c}$ that minimises \noun{Euclid}'s
quadratic distance $d_{E}^{2}(\mathbf{c},\mathbf{r})=\Vert\mathbf{r}-\mathbf{c}\Vert^{2}$
needs to be found.$ $
\end{anm}
For the independent channels of Remark~\ref{anm:BSC_AWGN} the ML
decisions can be computed (see Appendix~\ref{sub:Trellis-Based-Algorithms})
in the code trellis by the \noun{Viterbi }or the BCJR algorithm. However,
due to the generally large trellis complexity of the overall code
these algorithms do there (practically) not apply. 

On the other hand one may compute the {}``uncoded'' word probabilities
\begin{equation}
P(\mathbf{s}|\mathbf{r})\propto P(\mathbf{r}|\mathbf{s})\left\langle \mathbf{s}\in\mathbb{S}\right\rangle ,\label{eq:prsnocode}\end{equation}
and for small constituent trellis complexities the constituent code
word probabilities \[
P^{(l)}(\mathbf{s}|\mathbf{r}):=P_{\mathbf{C}^{(l)}|\mathbf{R}}(\mathbf{s}|\mathbf{r})=\frac{P(\mathbf{r}|\mathbf{s})\cdot\left\langle \mathbf{s}\in\mathbb{C}^{(l)}\right\rangle }{\sum_{\mathbf{s}'\in\mathbb{S}}P(\mathbf{r}|\mathbf{s}')\cdot\left\langle \mathbf{s}'\in\mathbb{C}^{(l)}\right\rangle }\propto P(\mathbf{r}|\mathbf{s})\cdot\left\langle \mathbf{s}\in\mathbb{C}^{(l)}\right\rangle \]
 for $l=1,2$ with $\mathbb{S}:=\mathbb{E}(\mathbf{S})$ the set of
all words. This is interesting as the overall code word distribution
\[
P^{(a)}(\mathbf{s}|\mathbf{r}):=P_{\mathbf{C}^{(a)}|\mathbf{R}}(\mathbf{s}|\mathbf{r})\propto P(\mathbf{r}|\mathbf{s})\cdot\left\langle \mathbf{s}\in\mathbb{C}^{(a)}\right\rangle \]
can be computed out of $P^{(l)}(\mathbf{s}|\mathbf{r})$ and $P(\mathbf{s}|\mathbf{r})$:
It holds with Definition~\ref{Def:The-dual-coupling} that $\mathbb{C}^{(a)}=\mathbb{C}^{(1)}\cap\mathbb{C}^{(2)}$
and thus

\[
P^{(1)}(\mathbf{s}|\mathbf{r})\cdot P^{(2)}(\mathbf{s}|\mathbf{r})\propto\left(P(\mathbf{r}|\mathbf{s})\right)^{2}\cdot\left\langle \mathbf{s}\in\mathbb{C}^{(1)}\right\rangle \cdot\left\langle \mathbf{s}\in\mathbb{C}^{(2)}\right\rangle =\left(P(\mathbf{r}|\mathbf{s})\right)^{2}\cdot\left\langle \mathbf{s}\in\mathbb{C}^{(a)}\right\rangle ,\]
which gives with (\ref{eq:prsnocode}) that \begin{equation}
P^{(a)}(\mathbf{s}|\mathbf{r})\propto\frac{P^{(1)}(\mathbf{s}|\mathbf{r})P^{(2)}(\mathbf{s}|\mathbf{r})}{P(\mathbf{s}|\mathbf{r})}.\label{eq:psrnoside}\end{equation}

If the constituent word probabilities are all known then optimal decoding
decisions can be taken. I.e., one can compute the ML word decision
by\begin{equation}
\hat{\mathbf{c}}^{(a)}=\mbox{arg}\max_{\mathbf{s}\in\mathbb{S}}\frac{P^{(1)}(\mathbf{s}|\mathbf{r})P^{(2)}(\mathbf{s}|\mathbf{r})}{P(\mathbf{s}|\mathbf{r})}\label{eq:DecMLw2}\end{equation}
or the ML symbol decisions by\begin{equation}
\bar{c}_{i}^{(a)}=\mbox{arg}\max_{x\in\mathbb{B}}P_{C_{i}}^{(a)}(x|\mathbf{r})=\mbox{arg}\max_{x\in\mathbb{B}}\sum_{\mathbf{s}\in\mathbb{S},s_{i}=x}P^{(a)}(\mathbf{s}|\mathbf{r})=\mbox{arg}\max_{x\in\mathbb{B}}\sum_{\mathbf{s}\in\mathbb{S}_{i}(x)}\frac{P^{(1)}(\mathbf{s}|\mathbf{r})P^{(2)}(\mathbf{s}|\mathbf{r})}{P(\mathbf{s}|\mathbf{r})}\label{eq:DecMLS2}\end{equation}
with $\mathbb{S}_{i}(x):=\{\mathbf{s}\in\mathbb{S}:s_{i}=x\}.$

Decoding decisions may therefore be taken by the constituent probabilities.
However, one may by (\ref{eq:psrnoside}) only compute a value proportional
to each single word probability. The representation complexity of
the constituent word probability \emph{distribution} remains prohibitively
large. I.e., the decoding decisions by (\ref{eq:DecMLw2}) and (\ref{eq:DecMLS2})
do not reduce the overall complexity as all word probabilities have
to be jointly considered, which is equivalent to investigating the
complete code constraint.

\subsection{Belief Propagation}

The probabilities of the two constituent codes thus contain the complete
knowledge about the decoding problem. However, the constituent decoders
may not use this knowledge (with reasonable complexity) as then $2^{n}$
values need to be transferred. I.e., a realistic algorithm based on
the constituent probabilities should transfer only a small number
of parameters. 

In (loopy) belief propagation algorithm this is done by transmitting
only the constituently {}``believed'' symbol probabilities but to
repeat this several times. This algorithm is here shortly reviewed:
One first uses a transfer vector $\mathbf{w}^{(1)}$ to represent
the believed $P_{C_{i}}^{(1)}(x|\mathbf{r})$ of code $1$. This belief
representing transfer vector is then used together with $\mathbf{r}$
in the decoder of the other constituent code. I.e., a transfer vector
$\mathbf{w}^{(2)}$ is computed out of $P_{C_{i}}^{(2)}(x|\mathbf{r},\mathbf{w}^{(1)})$
that will then be reused for a new $\mathbf{w}^{(1)}$ by $P_{C_{i}}^{(1)}(x|\mathbf{r},\mathbf{w}^{(2)})$
and so forth. The algorithm is stopped if the beliefs do not change
any further and a decoding decision is emitted.

The beliefs $P_{C_{i}}^{(h)}(x|\mathbf{r},\mathbf{w}^{(l)})$ for
$l,h\in\{1,2\}$ and $l\neq h$ are obtained by\[
P(\mathbf{r},\mathbf{w}^{(l)}|\mathbf{s})=P(\mathbf{w}^{(l)}|\mathbf{s})P(\mathbf{r}|\mathbf{s}),\]
which is a in $\mathbf{w}$ and $\mathbf{r}$ independent representation.
Moreover, it is assumed that $s_{i}\in\mathbb{B}=\{-1,+1\}$ and that
\begin{equation}
P(\mathbf{w}^{(l)}|\mathbf{s})=\prod_{i=1}^{n}P(w_{i}^{(l)}|s_{i})\propto\exp_{2}(\sum_{i=1}^{n}w_{i}^{(l)}s_{i})=\exp_{2}(\mathbf{w}^{(l)}\mathbf{s}^{T})\label{eq:defwahrscheinlich}\end{equation}
are of the form of $P(\mathbf{r}|\mathbf{s})$ in Remark~\ref{anm:BSC_AWGN}. 

\begin{anm}
(\emph{Distributions} \emph{and Trellis}) Obviously many other choices
for $P(\mathbf{w}^{(l)}|\mathbf{s})$ exist. However, the again independent
description of the symbols $C_{i}=S_{i}$ in (\ref{eq:defwahrscheinlich})
leads to (see Appendix~\ref{sub:Trellis-Based-Algorithms}) the possibility
to use trellis based computations, i.e., the symbol probabilities
$P_{C_{i}}^{(l)}(x|\mathbf{r},\mathbf{w}^{(h)})$ can be computed
as before $P_{C_{i}}^{(l)}(x|\mathbf{r})$. 
\end{anm}
The transfer vector $\mathbf{w}^{(h)}$ for belief propagation for
given \textbf{$\mathbf{r}$ }and $\mathbf{w}^{(l)}$ with $l,h\in\{1,2\}$
and $h\neq l$ is defined by \begin{equation}
P_{C_{i}}(x|\mathbf{r},\mathbf{w}^{(1)},\mathbf{w}^{(2)})=P_{C_{i}}^{(h)}(x|\mathbf{r},\mathbf{w}^{(l)})\mbox{ for all }i.\label{eq:Definition-belief-propagation}\end{equation}
I.e., the beliefs under $\mathbf{r}$, $\mathbf{w}^{(1)}$, $\mathbf{w}^{(2)}$,
and no further set restriction are set such that they are equal to
the beliefs under $\mathbf{w}^{(l)}$, $\mathbf{r}$, and the knowledge
of the set restriction of the $h$-th constituent code. This is always
possible as shown below.

\begin{anm}
\label{anm:(Notation)-m_w}(\emph{Notation}) To simplify the notation
we set in the sequel\[
\mathbf{m}=(\mathbf{r},\mathbf{w}^{(1)},\mathbf{w}^{(2)}),\,\mathbf{m}^{(1)}=(\mathbf{r},\mathbf{w}^{(2)}),\,\mathbf{m}^{(2)}=(\mathbf{r},\mathbf{w}^{(1)}),\]
and often $\mathbf{w}^{(0)}:=\mathbf{r}$.
\end{anm}
For the \emph{uncoded} beliefs $P_{C_{i}}(x|\mathbf{m})$ it is again
assumed that the information and belief carrying $\mathbf{r}$, $\mathbf{w}^{(1)}$
and $\mathbf{w}^{(2)}$ are independent, i.e., \begin{equation}
P(\mathbf{m}|\mathbf{c})=P(\mathbf{r},\mathbf{w}^{(1)},\mathbf{w}^{(2)}|\mathbf{c})=P(\mathbf{r}|\mathbf{c})P(\mathbf{w}^{(1)}|\mathbf{c})P(\mathbf{w}^{(2)}|\mathbf{c}).\label{eq:Overall-proba-uncod}\end{equation}

The computation of $\mathbf{w}^{(h)}$ for given $\mathbf{w}^{(l)}$
is then simple as the independence assumptions (\ref{eq:defwahrscheinlich})
and (\ref{eq:Overall-proba-uncod}) give that \[
P_{C_{i}}(x|\mathbf{m})=P_{C_{i}|\mathbf{R}}(x|r_{i}+w_{i}^{(1)}+w_{i}^{(2)}).\]
Moreover, the definition of the $\mathbf{w}^{(l)}$ is simplified
by the use of logarithmic probability ratios \[
L_{i}(\mathbf{m})=\frac{1}{2}\log_{2}\frac{P_{C_{i}}(+1|\mathbf{m})}{P_{C_{i}}(-1|\mathbf{m})}\mbox{ and }L_{i}^{(l)}(\mathbf{m}^{(l)})=\frac{1}{2}\log_{2}\frac{P_{C_{i}}^{(l)}(+1|\mathbf{m}^{(l)})}{P_{C_{i}}^{(l)}(-1|\mathbf{m}^{(l)})}\]
for $l=1,2$. This representation is handy for the computations as
(\ref{eq:defwahrscheinlich}) directly gives \[
L_{i}(\mathbf{m})=r_{i}+w_{i}^{(1)}+w_{i}^{(2)}\]
and thus that Equation~(\ref{eq:Definition-belief-propagation})
is equivalent to \begin{equation}
w_{i}^{(l)}=L_{i}^{(l)}(\mathbf{m}^{(l)})-r_{i}-w_{i}^{(h)}\mbox{ for }l\neq h\mbox{ and all }i.\label{eq:Iteration-Belief-Prob}\end{equation}
This equation can be used as an \emph{iteration} rule such that the
uncoded beliefs are subsequently updated by the constituent beliefs.
The transfer vectors $\mathbf{w}^{(1)}$ and $\mathbf{w}^{(2)}$ are
thereby via (\ref{eq:Iteration-Belief-Prob}) iteratively updated.
The following definition further simplifies the notation.

\begin{Def}
(Extrinsic Symbol Probability) The extrinsic symbol probability of
code $l$ is\[
\breve{P}_{C_{i}}^{(l)}(x|\mathbf{m}^{(l)})\propto P_{C_{i}}^{(l)}(x|\mathbf{m}^{(l)})\exp_{2}(-x(w_{i}^{(h)}-r_{i}))\mbox{ for }h\neq l.\]

\end{Def}
\begin{algorithm}
\begin{enumerate}
\item Set $\mathbf{w}^{(1)}=\mathbf{w}^{(2)}=\mathbf{0}$, $l=1$, and $h=2$.
\item Swap $l$ and $h$.
\item Set $\mathbf{w}^{(l)}=\breve{\mathbf{L}}^{(l)}(\mathbf{m}^{(l)})$. 
\item If $\mathbf{w}^{(h)}\neq\breve{\mathbf{L}}^{(h)}(\mathbf{m}^{(h)})$
then go to Step 2.
\item Set $\hat{c}_{i}=\mbox{sign}(r_{i}+w_{i}^{(1)}+w_{i}^{(2)})\mbox{ for all }i$.
\end{enumerate}
\caption{\label{alg:Loopy-Belief-Propagation}Loopy Belief Propagation}

\end{algorithm}

The extrinsic symbol probabilities are by (\ref{eq:defwahrscheinlich})
independent of $w_{i}^{(l)}$ for $l=1,2$ and $r_{i}$, i.e., they
depend only on belief and information carrying $w_{j}^{(l)}$ and
$r_{j}$ from with $j\neq i$ \emph{other} or {}``extrinsic'' symbol
positions. Moreover, one directly obtains the extrinsic logarithmic
probability ratios \begin{equation}
\breve{L}_{i}^{(l)}(\mathbf{m}^{(l)}):=\frac{1}{2}\log_{2}\frac{P_{C_{i}}^{(l)}(+1|\mathbf{m}^{(l)})}{P_{C_{i}}^{(l)}(-1|\mathbf{m}^{(l)})}-r_{i}-w_{i}^{(h)}=L_{i}^{(l)}(\mathbf{m}^{(l)})-r_{i}-w_{i}^{(h)}\mbox{ for }l\neq h.\label{eq:Extinsic-notation-L}\end{equation}
With Equation~(\ref{eq:Iteration-Belief-Prob}) this gives the iteration
rule \[
w_{i}^{(l)}=\breve{L}_{i}^{(l)}(\mathbf{m}^{(l)})\]
and thus Algorithm~\ref{alg:Loopy-Belief-Propagation}. Note that
one generally uses an alternative, less stringent stopping criterion
in Step~4 of the algorithm. 

If the algorithm converges then one obtains that \[
r_{i}+w_{i}^{(1)}+w_{i}^{(2)}=L_{i}^{(2)}(\mathbf{r},\mathbf{w}^{(1)})=L_{i}^{(1)}(\mathbf{r},\mathbf{w}^{(2)})\]
and \begin{equation}
\hat{c}_{i}=\mbox{sign}(L_{i}(\mathbf{r})+\breve{L}_{i}^{(1)}(\mathbf{m}^{(1)})+\breve{L}_{i}^{(2)}(\mathbf{m}^{(2)}))\label{eq:code-dec-belief}\end{equation}
with $L_{i}(\mathbf{r})=r_{i}.$ This is a rather intuitive form of
the fixed point of iterative belief propagation. The decoding decision
$\hat{c}_{i}$ is defined by the sum of the (representations of the)
channel information $r_{i}$ and the extrinsic constituent code beliefs
$\breve{L}_{i}^{(l)}(\mathbf{m}^{(l)})$. 

\begin{anm}
(\emph{Performance}) If the algorithm converges then simulations show
that the decoding decision is usually good. By density evolution~\cite{richardson01capacity}
or extrinsic information transfer charts~\cite{brink99} the convergence
of iterative belief propagation is further investigated. These approaches
evaluate, which constituent codes are suitable for iterative belief
propagation. This approach and simulations show that only rather weak
codes should be employed for good convergence properties. This indicates
that the chosen transfer is often too optimistic about its believed
decisions. 
\end{anm}

\subsection{Discrimination}

The belief propagation algorithm uses only knowledge about the constituent
codes represented by $\mathbf{w}^{(l)}.$ In this section we aim at
increasing the transfer complexity by adding more variables and hope
to obtain thereby a better representation of the overall information
and thus  an improvement over the propagation of only symbol beliefs. 

Reconsider first the additional belief representation $\mathbf{w}^{(l)}$
given by the distributions $P(\mathbf{s}|\mathbf{w}^{(1)})$ and $P(\mathbf{s}|\mathbf{w}^{(2)})$
used for belief propagation. The overall distributions are\begin{equation}
\begin{array}{rcccl}
P(\mathbf{s}|\mathbf{m}) & = & P(\mathbf{s}|\mathbf{r},\mathbf{w}^{(1)},\mathbf{w}^{(2)}) & \propto & P(\mathbf{r}|\mathbf{s})P(\mathbf{w}^{(1)}|\mathbf{s})P(\mathbf{w}^{(2)}|\mathbf{s})\\
P^{(1)}(\mathbf{s}|\mathbf{m}^{(1)}) & = & P^{(1)}(\mathbf{s}|\mathbf{r},\mathbf{w}^{(2)}) & \propto & P(\mathbf{r}|\mathbf{s})P(\mathbf{w}^{(2)}|\mathbf{s})\\
P^{(2)}(\mathbf{s}|\mathbf{m}^{(2)}) & = & P^{(2)}(\mathbf{s}|\mathbf{r},\mathbf{w}^{(1)}) & \propto & P(\mathbf{r}|\mathbf{s})P(\mathbf{w}^{(1)}|\mathbf{s}).\end{array}\label{eq:Independent}\end{equation}

The following lemma first gives that these additional beliefs do not
change the computation of the overall word probabilities. 

\begin{lem}
\label{lem:DisProb}It holds for all $\mathbf{w}^{(1)},\mathbf{w}^{(2)}$
that \[
P^{(a)}(\mathbf{s}|\mathbf{r})\propto\frac{P^{(1)}(\mathbf{s}|\mathbf{m}^{(1)})P^{(2)}(\mathbf{s}|\mathbf{m}^{(2)})}{P(\mathbf{s}|\mathbf{m})}.\]

\end{lem}
\begin{proof}
A direct computation of the equation with (\ref{eq:Independent})
gives as for (\ref{eq:psrnoside}) equality. The terms that depend
on $\mathbf{w}^{(l)}$ vanish by the independence assumption (\ref{eq:defwahrscheinlich}).
\end{proof}
To increase the transfer complexity now additional parameters are
added to $\mathbf{s}$. This first seems counter intuitive as no new
knowledge is added. However, with Lemma~\ref{lem:DisProb} the same
holds true for the belief carrying $\mathbf{w}^{(l)}$ and optimal
decoding.

\begin{Def}
(Word uncertainty) \label{Def:The-Word-Information}The uncertainty
augmented word probability $P^{(h)}(\mathbf{s},\mathbf{u}|\mathbf{m}^{(h)})$
is\[
P^{(h)}(\mathbf{s},\mathbf{u}|\mathbf{m}^{(h)}):=P^{(h)}(\mathbf{s}|\mathbf{m}^{(h)})\prod_{l=0}^{2}\left\langle u_{l}=\mathbf{w}^{(l)}\mathbf{s}^{T}\right\rangle \]
with $\mathbf{u}=\mathbf{u}(\mathbf{s})=(u_{0},u_{1},u_{2}).$
\end{Def}
This definition naturally extends to $P(\mathbf{s},\mathbf{u}|\mathbf{m})$
and to $P^{(a)}(\mathbf{s},\mathbf{u}|\mathbf{r})$. 

\begin{anm}
\label{anm:(Notation)-abhvonm}(\emph{Notation}) The notation of $P^{(a)}(\mathbf{s},\mathbf{u}|\mathbf{r})$
does not reflect the dependency on $\mathbf{m}$. The same holds true
for $P^{(l)}(\mathbf{s},\mathbf{u}|\mathbf{m}^{(l)})$ etc. A complete
notation is for example $P(\mathbf{s},\mathbf{u}|\mathbf{m}\Vert\mathbf{r})$
or $P^{(l)}(\mathbf{s},\mathbf{u}|\mathbf{m}\Vert\mathbf{m}^{(l)})$.
To maintain readability this dependency will \emph{not} be explicitly
stated in the sequel.
\end{anm}
Under the assumption that code words with the same $\mathbf{u}$ do
not need to be distinguished one obtains the following definition. 

\begin{Def}
\label{Def:(Discriminated-Word-Probabilities)}(Discriminated Distribution)
The distribution of $\mathbf{u}$ discriminated by $\mathbf{m}$ is \[
P^{\otimes}(\mathbf{u}|\mathbf{m})\propto\frac{P^{(1)}(\mathbf{u}|\mathbf{m}^{(1)})P^{(2)}(\mathbf{u}|\mathbf{m}^{(2)})}{P(\mathbf{u}|\mathbf{m})}\]
with $\sum_{\mathbf{u}\in\mathbb{U}}P^{\otimes}(\mathbf{u}|\mathbf{m})=1$,\[
P^{(l)}(\mathbf{u}|\mathbf{m}^{(l)})=\sum_{\mathbf{s}\in\mathbb{S}}P^{(l)}(\mathbf{s},\mathbf{u}|\mathbf{m}^{(l)}),\]
and $\mathbb{U}=\mathbb{E}(\mathbf{U})=\{\mathbf{u}:u_{l}=\mathbf{w}^{(l)}\mathbf{s}^{T}\,\forall l\mbox{ and }\mathbf{s}\in\mathbb{S}\}$.
\end{Def}
\begin{anm}
\label{anm:(Discrimination)}(\emph{Discrimination}) Words $\mathbf{s}$
with the same $\mathbf{u}$ are not  distinguished. As $\mathbf{m}$
and $\mathbf{s}$ define $\mathbf{u}$ the \emph{discrimination} of
words is steered by $\mathbf{m}$. The variables $u_{l}$ are then
used to relate to the \emph{distances} $\Vert\mathbf{c}-\mathbf{w}^{(l)}\Vert^{2}$
(see Remark~\ref{anm:BSC_AWGN}). Words that do no not share the
same distances are discriminated. The choice of $\mathbf{u}$ and
(\ref{eq:defwahrscheinlich}) is natural as all code words with the
same $\mathbf{u}$ have the same probability, i.e., that\begin{equation}
P^{(l)}(\mathbf{s},\mathbf{u}|\mathbf{m}^{(l)})\propto\exp_{2}(\sum_{{\scriptstyle k=0},k\neq l}^{2}u_{k})\cdot\prod_{j=0}^{2}\left\langle u_{j}=\mathbf{w}^{(j)}\mathbf{s}^{T}\right\rangle \cdot\left\langle \mathbf{s}\in\mathbb{C}^{(l)}\right\rangle \label{eq:Plsuml}\end{equation}
and similar for $P(\mathbf{s},\mathbf{u}|\mathbf{m})$ and $P^{(a)}(\mathbf{s},\mathbf{u}|\mathbf{r})$.
Generally it holds that $\mathbf{u}$ is via\[
\sum_{{\scriptstyle k=0}}^{2}u_{k}=\mathrm{K}+\log_{2}P(\mathbf{s}|\mathbf{m})=\mathrm{K}-H(\mathbf{s|}\mathbf{m})\]
(with $\mathrm{K}$ some constant)  related to the uncertainty $H(\mathbf{s|}\mathbf{m})$.
Note that any map of $\mathbf{s}$ on some $\mathbf{u}$ will define
some discrimination. However, we will here only consider the correlation
map, respectively the discrimination of the information theoretic
word uncertainties.
\end{anm}
 In the same way one obtains the much more interesting (uncertainty)
discriminated symbol probabilities. 

\begin{Def}
\label{Def:(Discriminated-Symbol-Probabilities)}(Discriminated Symbol
Probabilities) The symbol probabilities discriminated by $\mathbf{m}$
are \begin{equation}
P_{C_{i}}^{\otimes}(x|\mathbf{m})=\sum_{\mathbf{u}\in\mathbb{U}}P_{C_{i}}^{\otimes}(x,\mathbf{u}|\mathbf{m})\propto\sum_{\mathbf{u}\in\mathbb{U}}\frac{P_{C_{i}}^{(1)}(x,\mathbf{u}|\mathbf{m}^{(1)})P_{C_{i}}^{(2)}(x,\mathbf{u}|\mathbf{m}^{(2)})}{P_{C_{i}}(x,\mathbf{u}|\mathbf{m})}\label{eq:discriminated-symbol-probabilities}\end{equation}
with \[
P_{C_{i}}^{(l)}(x,\mathbf{u}|\mathbf{m}^{(l)})=\sum_{\mathbf{s}\in\mathbb{S}_{i}(x)}P^{(l)}(\mathbf{s},\mathbf{u}|\mathbf{m}^{(l)})\propto\sum_{\mathbf{s}\in\mathbb{C}^{(l)},s_{i}=x}P(\mathbf{s},\mathbf{u}|\mathbf{m}^{(l)}).\]

\end{Def}
\begin{anm}
\label{anm:(Independence)-w_i^(l)}(\emph{Independence}) Note that
$P_{C_{i}}^{\otimes}(x|\mathbf{m})$ is by (\ref{eq:defwahrscheinlich})
independent of both $w_{i}^{(l)}$.
\end{anm}
The discriminated symbol probabilities may be considered as  \emph{commonly}
believed symbol probabilities under discriminated word uncertainties. 

To obtain a first intuitive understanding of this fact we relate $P_{C_{i}}^{\otimes}(x|\mathbf{m})$
to the more accessible constituent symbol probabilities $P_{C_{i}}^{(l)}(x|\mathbf{m}^{(l)})$.
It holds by \noun{Bayes}' theorem that \[
P_{C_{i}}^{\otimes}(x|\mathbf{m})\propto\frac{P_{C_{i}}^{(1)}(x|\mathbf{m}^{(1)})P_{C_{i}}^{(2)}(x|\mathbf{m}^{(2)})}{P_{C_{i}}(x|\mathbf{m})}P_{C_{i}}^{{\scriptscriptstyle \,\boxtimes}}(x|\mathbf{m})\]
with (abusing notation as this is not a probability) \begin{equation}
P_{C_{i}}^{{\scriptscriptstyle \boxtimes}}(x|\mathbf{m})\propto\sum_{\mathbf{u}\in\mathbb{U}}\frac{P_{C_{i}}^{(1)}(\mathbf{u}|x,\mathbf{m}^{(1)})P_{C_{i}}^{(2)}(\mathbf{u}|x,\mathbf{m}^{(2)})}{P_{C_{i}}(\mathbf{u}|x,\mathbf{m})}.\label{eq:Pboxtimes}\end{equation}
In the logarithmic notation this gives\[
L_{i}^{\otimes}(\mathbf{m})=L_{i}^{(1)}(\mathbf{m}^{(1)})+L_{i}^{(2)}(\mathbf{m}^{(2)})-L_{i}(\mathbf{m})+L_{i}^{{\scriptscriptstyle \boxtimes}}(\mathbf{m})\]
or in the extrinsic notation of (\ref{eq:Extinsic-notation-L}) that\begin{equation}
L_{i}^{\otimes}(\mathbf{m})=\breve{L}_{i}^{(1)}(\mathbf{m}^{(1)})+\breve{L}_{i}^{(2)}(\mathbf{m}^{(2)})+L_{i}(\mathbf{r})+L_{i}^{{\scriptscriptstyle \boxtimes}}(\mathbf{m}).\label{eq:logotimesextr}\end{equation}
Note first the similarity with (\ref{eq:code-dec-belief}). One  has
again a sum of the extrinsic beliefs, however, an additional value
$L_{i}^{{\scriptscriptstyle \boxtimes}}(\mathbf{m})$ is added, which
is by Remark~\ref{anm:(Independence)-w_i^(l)} necessarily independent
of $w_{i}^{(l)}$ and $l=1,2.$ Overall the common belief joins the
two constituent beliefs together with a {}``distance'' correction
term.

Below we show that this new common belief is -- under again practically
prohibitively high complexity -- just the real overall {}``belief'',
i.e., the correct symbol probabilities obtained by optimal symbol
decoding. 

\begin{Def}
\label{Def:(Globally-Maximal-Discriminator)}(Globally Maximal Discriminator)
The discriminator $\mathbf{m}$ is globally maximal  (for $\mathbf{S}$)
if $|\mathbb{S}(\mathbf{u}|\mathbf{m})|=1$ for all $\mathbf{u}\in\mathbb{U}$.
I.e., for globally maximal discriminators exists a one-to-one correspondence
between $\mathbf{s}$ and $\mathbf{u}$ and thus $|\mathbb{S}|=|\mathbb{U}|$. 
\end{Def}
\begin{lem}
\label{lem:MaxDiskp(s)}\label{lem:globally-maximal=00003Doptsymbols}For
a globally maximal discriminator $\mathbf{m}$ it holds that \[
P^{\otimes}(\mathbf{u}|\mathbf{m})=P^{(a)}(\mathbf{u}|\mathbf{r})\mbox{ and }P_{C_{i}}^{\otimes}(x|\mathbf{m})=P_{C_{i}}^{(a)}(x|\mathbf{r}).\]
 I.e., the by $\mathbf{m}$ discriminated symbol probabilities are
correct.
\end{lem}
\begin{proof}
Lemma~\ref{lem:DisProb} and Definition~\ref{Def:The-Word-Information}
give \[
P^{(a)}(\mathbf{s},\mathbf{u}|\mathbf{r})\propto\frac{P^{(1)}(\mathbf{s},\mathbf{u}|\mathbf{m}^{(1)})P^{(2)}(\mathbf{s},\mathbf{u}|\mathbf{m}^{(2)})}{P(\mathbf{s},\mathbf{u}|\mathbf{m})}\]
as $\mathbf{u}$ follows directly from $\mathbf{s}$. For a globally
maximal discriminator $\mathbf{m}$ exists a one-to-one correspondence
between $\mathbf{s}$ and $\mathbf{u}$ This gives that one can omit
for any probability either $\mathbf{u}$ or $\mathbf{s}$. This proves
the optimality of the discriminated distribution. 

For the overall symbol probabilities holds \[
P_{C_{i}}^{(a)}(x|\mathbf{r})=\sum_{\mathbf{s}\in\mathbb{S}_{i}(x)}P^{(a)}(\mathbf{s}|\mathbf{r})=\sum_{\mathbf{s}\in\mathbb{S}_{i}(x)}\frac{P^{(1)}(\mathbf{s},\mathbf{u}|\mathbf{m}^{(1)})P^{(2)}(\mathbf{s},\mathbf{u}|\mathbf{m}^{(2)})}{P(\mathbf{s},\mathbf{u}|\mathbf{m})}.\]
With $P_{C_{i}}^{(l)}(x,\mathbf{s},\mathbf{u}|\mathbf{m}^{(l)})=P^{(l)}(\mathbf{s},\mathbf{u}|\mathbf{m}^{(l)})$
for $\mathbf{s}\in\mathbb{S}_{i}(x)$ and $P_{C_{i}}^{(l)}(x,\mathbf{s},\mathbf{u}|\mathbf{m}^{(l)})=0$
for $\mathbf{s}\not\in\mathbb{S}_{i}(x)$ the right hand side becomes
\[
\sum_{\mathbf{s}\in\mathbb{S}_{i}(x)}\frac{P^{(1)}(\mathbf{s},\mathbf{u}|\mathbf{m}^{(1)})P^{(2)}(\mathbf{s},\mathbf{u}|\mathbf{m}^{(2)})}{P(\mathbf{s},\mathbf{u}|\mathbf{m})}=\sum_{\mathbf{s}\in\mathbb{S}}\frac{P_{C_{i}}^{(1)}(x,\mathbf{s},\mathbf{u}|\mathbf{m}^{(1)})P_{C_{i}}^{(2)}(x,\mathbf{s},\mathbf{u}|\mathbf{m}^{(2)})}{P_{C_{i}}(x,\mathbf{s},\mathbf{u}|\mathbf{m})}.\]
By the one-to-one correspondence one can replace the sum over $\mathbf{s}$
by a sum over $\mathbf{u}$ to obtain \[
P_{C_{i}}^{(a)}(x|\mathbf{r})=\sum_{\mathbf{u}\in\mathbb{U}}\frac{P_{C_{i}}^{(1)}(x,\mathbf{s},\mathbf{u}|\mathbf{m}^{(1)})P_{C_{i}}^{(2)}(x,\mathbf{s},\mathbf{u}|\mathbf{m}^{(2)})}{P_{C_{i}}(x,\mathbf{s},\mathbf{u}|\mathbf{m})},\]
which is ($\mathbf{s}$ can be omitted due to  the one-to-one correspondence)
the optimality of the discriminated symbol probabilities.
\end{proof}
A globally maximal discriminator $\mathbf{m}$ thus solves the problem
of ML symbol by symbol decoding. Likewise by\[
\arg\max_{\mathbf{u}\in\mathbb{U}}P^{\otimes}(\mathbf{u}|\mathbf{m})=\arg\max_{\mathbf{u}\in\mathbb{U}}P^{(a)}(\mathbf{u}|\mathbf{r})=\arg\max_{\mathbf{s}\in\mathbb{S}}P^{(a)}(\mathbf{s}|\mathbf{r})\]
the problem of ML word decoding is solved (provided the one-to-one
correspondence of $\mathbf{u}$ and $\mathbf{s}$ can be easily inverted). 

This is not surprising as a globally maximal discriminator has by
the one-to-one correspondence of $\mathbf{s}$ and $\mathbf{u}$ the
discriminator complexity $|\mathbb{U}|=|\mathbb{S}|$. The transfer
complexity is then just the complexity of the optimal decoder based
on constituent probabilities. 

\begin{anm}
\label{anm:(Real-Valued-Discriminators)}(\emph{Globally} \emph{Maximal}
\emph{Discriminators}) The vector $\mathbf{m}=(\mathbf{r},\mathbf{w}^{(1)},\mathbf{0})$
and $w_{i}^{(1)}=2^{i}$ is an example of a globally maximal discriminator
as $u_{1}(\mathbf{s})=\sum_{i=1}^{n}s_{i}2^{i}$ is different for
all values of $\mathbf{s}$. I.e., there exists a one-to-one correspondence
between $\mathbf{s}$ and $\mathbf{u}$. Generally it is rather simple
to construct a globally maximal discriminator. E.g., the $\mathbf{r}$
received via an AWGN channel is usually already maximal discriminating:
The probability that two words $\mathbf{s}^{(1)},\mathbf{s}^{(2)}\in\mathbb{S}$
share the same real valued distance to the received word is generally
Zero.
\end{anm}

\section{Local  Discriminators}

In the last section the coupling of error correcting codes was reviewed
and different decoders were discussed. It was shown that an optimal
decoding is due to the large representation complexity practically
not feasible, but that a transfer of beliefs may lead to a good decoding
algorithm. A generalisation of this approach led to the concept of
discriminators and therewith to a new overall belief. The complexity
of the computation of this belief is depending on $|\mathbb{U}|$,
i.e., the number of different outcomes $\mathbf{u}$ of the discrimination.
Finally it was shown that the obtained overall belief leads to the
optimal overall decoding decision if the set is with $|\mathbb{U}|=|\mathbb{S}|$
maximally large. (However, then the overall decoding complexity is
not reduced.)

In this section we consider local discriminators with $|\mathbb{U}|\ll|\mathbb{S}|$.
Then only a limited number of values need to be transferred to compute
by (\ref{eq:discriminated-symbol-probabilities}) a new overall belief
$P_{C_{i}}^{\otimes}(x|\mathbf{m})$. These discriminated beliefs
$P_{C_{i}}^{\otimes}(x|\mathbf{m})$ may then be practically employed
to improve iterative decoding. To do so we first show that local discriminators
exist. 

\begin{exa}
\label{exa:Hard-not-maximal}The $\mathbf{r}$ obtained by a transmission
over a BSC is generally a local discriminator. The map $\mathcal{U}(\mathbf{r}):\mathbf{s}\mapsto\mathbf{u=}(u_{0},0,0)$
is then only dependent on the \noun{Hamming} distance $d_{H}(\mathbf{r},\mathbf{s})$,
i.e., \[
\mathcal{U}_{0}(\mathbf{r}):\mathbf{s}\mapsto u_{0}=\mathbf{r}\mathbf{s}^{T}=n-2d_{H}(\mathbf{r},\mathbf{s})\]
 and thus $\mathbb{U}=\mathbb{E}(U_{0})\subseteq\{-n,-n+2,...,n-2,n\}$,
which gives $|\mathbb{U}|\leq n+1.$ This furthermore gives that an
additional {}``hard decision'' choice of the $\mathbf{w}^{(l)}$
will continue to yield a local {}``\noun{Hamming}'' discriminator
$\mathbf{m}$. 
\end{exa}
To investigate local discrimination now reconsider the discriminated
distributions. With Remark~\ref{anm:(Discrimination)} one obtains
the following lemma.

\begin{lem}
\label{lem:Setsizedirect}The distributions of $\mathbf{u}$ given
$\mathbf{m}$ are\[
P(\mathbf{u}|\mathbf{m})\propto|\mathbb{S}(\mathbf{u}|\mathbf{m})|\exp_{2}(u_{0}+u_{1}+u_{2})\]
and\[
P^{(l)}(\mathbf{u}|\mathbf{m}^{(l)})\propto|\mathbb{C}^{(l)}(\mathbf{u}|\mathbf{m})|\exp_{2}(\sum_{{\scriptstyle k=0},k\neq l}^{2}u_{k})\]
where the sets $\mathbb{S}(\mathbf{u}|\mathbf{m})$ and $\mathbb{C}^{(l)}(\mathbf{u}|\mathbf{m})$
are defined by $\mathbb{M}(\mathbf{u}|\mathbf{m}):=\{\mathbf{s}\in\mathbb{M}:u_{l}=\mathbf{w}^{(l)}\mathbf{s}^{T}\,\forall l\}.$
\end{lem}
\begin{proof}
By (\ref{eq:Plsuml}) follows that the probability of all words $\mathbf{s}\in\mathbb{S}(\mathbf{u}|\mathbf{m})$
with the same $\mathbf{u}$ is equal and proportional to $\exp_{2}(\sum_{k=0}^{2}u_{k})$.
As $|\mathbb{S}(\mathbf{u}|\mathbf{m})|$ words are in $\mathbb{S}(\mathbf{u}|\mathbf{m})$
this gives the first equation. The second equation is obtained by
adding the code constraint.
\end{proof}
\begin{anm}
(\emph{Overall} \emph{Distribution})\label{anm:Overall Distribution}
In the same way follows (see Remark~\ref{anm:(Notation)-abhvonm})
that\begin{equation}
P^{(a)}(\mathbf{u}|\mathbf{r})\propto|\mathbb{C}^{(a)}(\mathbf{u}|\mathbf{m})|\exp_{2}(u_{0}).\label{eq:OverallsetProbs}\end{equation}
More general restrictions (see below) can always be handled by imposing
restrictions on the considered sets. One thus generally obtains for
the distributions of $\mathbf{u}$ a description via on $\mathbf{u}$
dependent sets sizes. 
\end{anm}
\begin{exa}
\label{exa:Set-sizes-maximal-discrim}With the concept of set sizes
Example~\ref{exa:Hard-not-maximal} is continued. Assume again that
the discriminator is given by $\mathbf{m}=(\mathbf{r},\mathbf{0},\mathbf{0})$.
In this case no discrimination takes place on $u_{1}$ and $u_{2}$
as one obtains $u_{1}=u_{2}=0$ for all $\mathbf{s}$. One first obtains
the overall distribution $P^{(a)}(\mathbf{u}|\mathbf{r})$ to be with
Remark~\ref{anm:Overall Distribution} the multiplication of $\exp_{2}(u_{0})$
with the distribution of the correlation $\mathbf{c}\mathbf{r}^{T}$
with $\mathbf{c}\in\mathbb{C}^{(a)}$ given by $|\mathbb{C}^{(a)}(\mathbf{u}|\mathbf{m})|$.
\\
Assume furthermore that the overall maximum likelihood decision $\hat{\mathbf{c}}^{(a)}$
is with \[
P^{(a)}(\hat{\mathbf{c}}^{(a)}|\mathbf{r})\to1\]
 \emph{distinguished}. This assumption gives that \[
P^{(a)}(\mathbf{u}|\mathbf{r})=P^{(a)}(u_{0}|\mathbf{r})\approx\begin{cases}
1 & \mbox{ for }u_{0}=u_{0}(\hat{\mathbf{c}}^{(a)})=n-2d_{H}(\mathbf{r},\hat{\mathbf{c}}^{(a)})\\
0 & \mbox{ else.}\end{cases}\]
I.e., $P^{(a)}(\mathbf{u}|\mathbf{r})$ consists of one peak. %
\begin{figure}[h]
\begin{centering}
\includegraphics{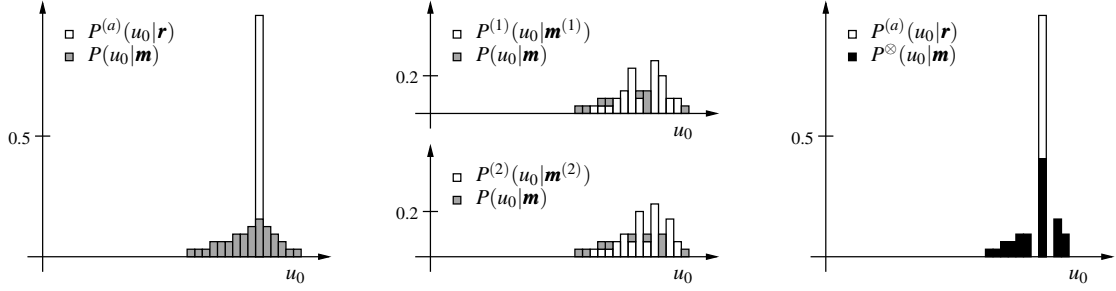}
\par\end{centering}

\caption{\label{fig:Hard-Decisions}Hard Decisions}

\end{figure}

For the other probabilities $P(u_{0}|\mathbf{m})$ and $P^{(l)}(u_{0}|\mathbf{m}^{(l)})$
 with Lemma~\ref{lem:Setsizedirect} again a multiplication of correlation
distributions with $\mbox{exp}_{2}(u_{0})$ is obtained. These distributions
will, however, due to the much larger spaces \[
|\mathbb{S}|\gg|\mathbb{C}^{(l)}|\gg|\mathbb{C}^{(a)}|\]
 usually not be in the form of a single peak. Other words with $u_{0}\geq\hat{\mathbf{c}}^{(a)}\mathbf{r}^{T}$
may appear. The same then holds true for~$P^{\otimes}(u_{0}|\mathbf{m})$.
These considerations are exemplary depicted in Figure~\ref{fig:Hard-Decisions}.
Note that the distributions can all be computed (see Appendix~\ref{sub:Trellis-Based-Algorithms})
in the constituent trellises.
\end{exa}
For a local discrimination a computation in the constituent trellises
produces by (\ref{eq:discriminated-symbol-probabilities}) symbol
probabilities $P_{C_{i}}^{\otimes}(x|\mathbf{m})$. In equivalence
to (loopy) belief propagation these probabilities should lead to the
definition of some $\mathbf{w}$ and thus to some iteration rule.
Before considering this approach we evaluate the quality of the discriminated
symbol probabilities.

\subsection{Typicality}

With Lemma~\ref{lem:Setsizedirect} one obtains that the discriminated
symbol probabilities defined by (\ref{eq:discriminated-symbol-probabilities})
are \begin{align}
P_{C_{i}}^{\otimes}(x|\mathbf{m}) & \propto\sum_{\mathbf{u}\in\mathbb{U}}\frac{|\mathbb{C}_{i}^{(1)}(x,\mathbf{u}|\mathbf{m})|\exp_{2}(u_{0}+u_{2})|\mathbb{C}_{i}^{(2)}(x,\mathbf{u}|\mathbf{m})|\exp_{2}(u_{0}+u_{1})}{|\mathbb{S}_{i}(x,\mathbf{u}|\mathbf{m})|\exp_{2}(u_{0}+u_{1}+u_{2})}\nonumber \\
 & =\sum_{\mathbf{u}\in\mathbb{U}}\frac{|\mathbb{C}_{i}^{(1)}(x,\mathbf{u}|\mathbf{m})||\mathbb{C}_{i}^{(2)}(x,\mathbf{u}|\mathbf{m})|}{|\mathbb{S}_{i}(x,\mathbf{u}|\mathbf{m})|}\exp_{2}(u_{0})\label{eq:axelsmeckerei}\end{align}
with the sets $\mathbb{C}_{i}^{(l)}(x,\mathbf{u}|\mathbf{m})$ defined
by $\mathbf{s}\in\mathbb{C}^{(l)}$ and $s_{i}=x$. Hence, $P_{C_{i}}^{\otimes}(x|\mathbf{m})$
only depends on the discriminated set sizes $\mathbb{C}_{i}^{(l)}(x,\mathbf{u}|\mathbf{m})$,
$\mathbb{S}_{i}(x,\mathbf{u}|\mathbf{m})$, and the word probabilities
$P(\mathbf{s}|\mathbf{r})\propto\exp_{2}(u_{0}(\mathbf{s}))$.

The discriminated symbol probabilities should approximate the overall
probabilities, i.e., \[
P_{C_{i}}^{\otimes}(x|\mathbf{m})\approx P_{C_{i}}^{(a)}(x|\mathbf{r}).\]
With Remark~\ref{anm:Overall Distribution} and (\ref{eq:axelsmeckerei})
this approximation is surely good if \begin{equation}
\frac{|\mathbb{C}_{i}^{(1)}(x,\mathbf{u}|\mathbf{m})||\mathbb{C}_{i}^{(2)}(x,\mathbf{u}|\mathbf{m})|}{|\mathbb{S}_{i}(x,\mathbf{u}|\mathbf{m})|}\approx|\mathbb{C}_{i}^{(a)}(x,\mathbf{u}|\mathbf{m})|.\label{eq:sizesmarginalapprox}\end{equation}
Intuitively, the approximation thus uses the knowledge how many words
of the same correlation values $\mathbf{u}$ and decision $c_{i}=x$
are in both codes simultaneously. Moreover, depending on the discriminator
$\mathbf{m}$ the quality of this approximation will change.

An average consideration of the approximations (\ref{eq:sizesmarginalapprox})
is related to the following lemma.

\begin{lem}
\label{lem:Dual_independence}If the duals of the (linear) constituent
codes do not share common words but the zero word then\begin{equation}
|\mathbb{C}^{(1)}||\mathbb{C}^{(2)}|=|\mathbb{S}||\mathbb{C}^{(a)}|.\label{eq:meansizesareequal}\end{equation}

\end{lem}
\begin{proof}
With Definition~\ref{Def:The-dual-coupling} and by assumption linearly
independent $\mathbf{H}^{(1)}$ and $\mathbf{H}^{(2)}$ it holds that
the dual code dimension of the coupled code is just the sum of the
dual code dimension of the constituent codes, i.e.,  \[
n-k=(n-k^{(1)})+(n-k^{(2)}).\]
This is equivalent to $k^{(1)}+k^{(2)}=n+k$ and thus  the statement
of the lemma. 
\end{proof}
This lemma extends to the constrained set sizes $|\mathbb{C}_{i}^{(l)}(x)|$
as used in (\ref{eq:sizesmarginalapprox}). The approximations are
thus in the \emph{mean} correct. 

For random coding and independently chosen $\mathbf{m}=(\mathbf{r},\mathbf{w}^{(1)},\mathbf{w}^{(2)})$
this consideration can be put into a more precise form. 

\begin{lem}
\label{lem:AsympGleichZufSets}For random (long) codes $\mathbb{C}^{(1)}$
and $\mathbb{C}^{(2)}$ and independently chosen $\mathbf{m}$ holds
the asymptotic equality \begin{equation}
|\mathbb{C}_{i}^{(a)}(x,\mathbf{u}|\mathbf{m})|\asymp\frac{|\mathbb{C}_{i}^{(1)}(x,\mathbf{u}|\mathbf{m})||\mathbb{C}_{i}^{(2)}(x,\mathbf{u}|\mathbf{m})|}{|\mathbb{S}_{i}(x,\mathbf{u}|\mathbf{m})|}.\label{eq:asymEquaProbs}\end{equation}

\end{lem}
\begin{proof}
The probability of a random choice in $\mathbb{S}$ to be in $\mathbb{S}(\mathbf{u}|\mathbf{m})$
is just the fraction of the set sizes $|\mathbb{S}(\mathbf{u}|\mathbf{m})|$
and $|\mathbb{S}|$. 

For a random coupled code $|\mathbb{C}^{(a)}|$ the codewords are
a random subset of the set $|\mathbb{S}|$. For $|\mathbb{C}^{(a)}|\gg1$
the law of large numbers thus gives the asymptotic equality \begin{equation}
\frac{|\mathbb{C}_{i}^{(a)}(x,\mathbf{u}|\mathbf{m})|}{|\mathbb{C}^{(a)}|}\asymp\frac{|\mathbb{S}_{i}(x,\mathbf{u}|\mathbf{m})|}{|\mathbb{S}|}.\label{eq:asympgle1}\end{equation}
 The same holds true for the constituent codes \[
\frac{|\mathbb{C}_{i}^{(l)}(x,\mathbf{u}|\mathbf{m})|}{|\mathbb{C}^{(l)}|}\asymp\frac{|\mathbb{S}_{i}(x,\mathbf{u}|\mathbf{m})|}{|\mathbb{S}|}.\]
A multiplication of the equality of code 1 with the one of code 2
gives the asymptotic equivalence\begin{equation}
\frac{|\mathbb{S}|}{|\mathbb{C}^{(1)}||\mathbb{C}^{(2)}|}\frac{|\mathbb{C}_{i}^{(1)}(x,\mathbf{u}|\mathbf{m})||\mathbb{C}_{i}^{(2)}(x,\mathbf{u}|\mathbf{m})|}{|\mathbb{S}_{i}(x,\mathbf{u}|\mathbf{m})|}\asymp\frac{|\mathbb{S}_{i}(x,\mathbf{u}|\mathbf{m})|}{|\mathbb{S}|}.\label{eq:asmpgl2}\end{equation}
Combining (\ref{eq:asympgle1}) and (\ref{eq:asmpgl2}) then leads
to\[
\frac{|\mathbb{S}|}{|\mathbb{C}^{(1)}||\mathbb{C}^{(2)}|}\frac{|\mathbb{C}_{i}^{(1)}(x,\mathbf{u}|\mathbf{m})||\mathbb{C}_{i}^{(2)}(x,\mathbf{u}|\mathbf{m})|}{|\mathbb{S}_{i}(x,\mathbf{u}|\mathbf{m})|}\asymp\frac{|\mathbb{C}_{i}^{(a)}(x,\mathbf{u}|\mathbf{m})|}{|\mathbb{C}^{(a)}|}.\]
With (\ref{eq:meansizesareequal}) this is the statement of the lemma. 
\end{proof}
\begin{anm}
(\emph{Randomness}) The proof of the lemma indicates that the approximation
is rather good for code choices that are independent of $\mathbf{m}$.
I.e., perfect randomness of the codes is generally not needed. This
can be understood by the concept of random codes in information theory.
A random code is generally a good code. Conversely a good code should
not exhibit any structure, i.e., it behaves as a random code. 
\end{anm}

\subsection{Distinguished Words}

The received vector $\mathbf{r}$ is obtained from the channel \emph{and}
the encoding. The discriminator $\mathbf{m}$ is due to the dependent
$\mathbf{r}$ thus generally \emph{not} independent of the encoding.
This becomes directly clear by reconsidering Example~\ref{exa:Set-sizes-maximal-discrim}
and the assumptions that a \emph{distinguished} word $\hat{\mathbf{c}}^{(a)}$
with \[
P^{(a)}(\hat{\mathbf{c}}^{(a)}|\mathbf{r})\to1\]
exists. In this case the constituent distributions and thus likewise
the discriminator distribution $P^{\otimes}(\mathbf{u}|\mathbf{r})$
will be large in a region where a {}``typical'' number of errors
$\hat{t}$ occurred, i.e, $u_{0}=\mathbf{r}\mathbf{c}^{T}\approx n-2\hat{t}$.

For an independent $\mathbf{m}$, however, this would not be the case:
Then $P^{\otimes}(\mathbf{u}|\mathbf{m})$ would with Lemma~\ref{lem:AsympGleichZufSets}
be large in the vicinity of a typical minimal overall code word \emph{distance}.
This distance is generally larger than the expected number of errors
$\hat{t}$ under a distinguished word. Hence, $P^{\otimes}(\mathbf{u}|\mathbf{m})$
would then be large at a smaller $u_{0}$ than under a dependent $\mathbf{m}$.

\begin{anm}
(\emph{Channel} \emph{Capacity} \emph{and Typical Sets}) The existence
of a distinguished word is equivalent to assuming a long random code
of rate below capacity~\cite{Shannon_Inf}. The word sent is then
the only one in the typical set, i.e., it has a small distance to
$\mathbf{r}$. The other words of a random code will typically exhibit
a large distance to~$\mathbf{r}$.
\end{anm}
To describe single words one needs to describe \emph{how well} certain
environments in $\mathbf{u}$ given $\mathbf{m}$ are discriminated.
The precision of the approximation of $\mathbb{C}_{i}^{(a)}(x,\mathbf{u}|\mathbf{m})$
by (\ref{eq:sizesmarginalapprox}) hereby obviously depends on the
set size $|\mathbb{S}_{i}(x,\mathbf{u}|\mathbf{m})|$. This leads
to the following definition.

\begin{Def}
(Maximally Discriminated Region) The by $\mathbf{m}$ maximally discriminated
region \[
\mathbb{D}(\mathbf{m}):=\bigcup_{|\mathbb{S}(\mathbf{u}|\mathbf{m})|=1}\mathbb{S}(\mathbf{u}|\mathbf{m})\]
consists of all words $\mathbf{s}$ that uniquely define  $\mathbf{u}\mbox{ with }u_{l}=\mathbf{s}\mathbf{w}^{(l)T}\mbox{ for }l=0,1,2$.
\end{Def}
\begin{thm}
\label{thm:Asymptyp}For independent constituent codes and a by $\hat{\mathbf{c}}^{(a)}\in\mathbb{D}(\mathbf{m})$
maximally discriminated distinguished event is \[
P_{C_{i}}^{\otimes}(x|\mathbf{m})\asymp P_{C_{i}}^{(a)}(x|\mathbf{r})\mbox{ and }P^{\otimes}(\mathbf{u}|\mathbf{m})\asymp P^{(a)}(\mathbf{u}|\mathbf{r}).\]

\end{thm}
\begin{proof}
It holds with (\ref{eq:axelsmeckerei}) that \begin{equation}
P_{C_{i}}^{\otimes}(x|\mathbf{m})\propto\sum_{\mathbf{u}\in\mathbb{U}}\frac{|\mathbb{C}_{i}^{(1)}(x,\mathbf{u}|\mathbf{m})||\mathbb{C}_{i}^{(2)}(x,\mathbf{u}|\mathbf{m})|}{|\mathbb{S}_{i}(x,\mathbf{u}|\mathbf{m})|}\exp_{2}(u_{0}).\label{eq:otimesxMexp2u0}\end{equation}
For the distinguished event $\hat{\mathbf{c}}^{(a)}\in\mathbb{C}^{(a)}$
it follows that \[
\frac{|\mathbb{C}_{i}^{(1)}(x,\mathbf{u}(\hat{\mathbf{c}}^{(a)})|\mathbf{m})||\mathbb{C}_{i}^{(2)}(x,\mathbf{u}(\hat{\mathbf{c}}^{(a)})|\mathbf{m})|}{|\mathbb{S}_{i}(x,\mathbf{u}(\hat{\mathbf{c}}^{(a)})|\mathbf{m})|}=|\mathbb{C}_{i}^{(a)}(x,\mathbf{u}(\hat{\mathbf{c}}^{(a)})|\mathbf{m})|=1\,\mbox{ for }\,\hat{c}_{i}^{(a)}=x\]
as by assumption $\hat{\mathbf{c}}^{(a)}\in\mathbb{D}(\mathbf{m})$
is maximally discriminated, which gives by definition and $\hat{\mathbf{c}}^{(a)}\in\mathbb{C}^{(l)}$
for $l=1,2$ that\[
|\mathbb{S}_{i}(x,\mathbf{u}(\hat{\mathbf{c}}^{(a)})|\mathbf{m})|=|\mathbb{C}_{i}^{(1)}(x,\mathbf{u}(\hat{\mathbf{c}}^{(a)})|\mathbf{m})|=|\mathbb{C}_{i}^{(2)}(x,\mathbf{u}(\hat{\mathbf{c}}^{(a)})|\mathbf{m})|=1.\]
 I.e., the term with $\mathbf{u=}\mathbf{u}(\hat{\mathbf{c}}^{(a)})$
in (\ref{eq:otimesxMexp2u0}) is correctly estimated.

The other terms in (\ref{eq:otimesxMexp2u0}) represent non distinguished
words and can (with the assumption of independent constituent codes)
be considered to be independent of $\mathbf{m}$. This gives that
they can be assumed to be obtained by random coding. I.e., for\[
\mathbf{u}\neq\mathbf{u}(\hat{\mathbf{c}}^{(a)})\,\mbox{ with }\, u_{l}(\hat{\mathbf{c}}^{(a)})=\mathbf{w}^{(l)}\hat{\mathbf{c}}^{(a)T}\]
holds \[
\frac{|\mathbb{C}_{i}^{(1)}(x,\mathbf{u}|\mathbf{m})||\mathbb{C}_{i}^{(2)}(x,\mathbf{u}|\mathbf{m})|}{|\mathbb{S}_{i}(x,\mathbf{u}|\mathbf{m})|}\asymp|\mathbb{C}_{i}^{(a)}(x,\mathbf{u}|\mathbf{m})|\]
of Lemma~\ref{lem:AsympGleichZufSets}. Hence the other words are
(asymptotically) correctly estimated, too. 

Moreover, with (\ref{eq:otimesxMexp2u0}) one obtains for $\hat{\mathbf{c}}^{(a)}$
a probability value proportional to $\exp_{2}(\mathbf{r}\hat{\mathbf{c}}^{(a)T})$.
The other terms of (\ref{eq:otimesxMexp2u0}) are much smaller: An
independent random code typically does not exhibit code words of small
distance to $\mathbf{r}$. As the code rate is below capacity then
$P^{\otimes}(\mathbf{u}(\hat{\mathbf{c}}^{(a)})|\mathbf{m})$ exceeds
the sum of the probabilities of the other words. Asymptotically by~(\ref{eq:axelsmeckerei})
both the overall symbol probabilities and the overall distribution
of correlations follow. 
\end{proof}
\begin{anm}
\label{anm:(Distance)}(\emph{Distance}) Note that the multiplication
with $\exp_{2}(u_{0})$ in  (\ref{eq:otimesxMexp2u0}) excludes elements
that are not in the distinguished set ($\equiv$ with large distance
to $\mathbf{r}$). These words can -- as shown by information theory
-- not dominate (a random code) in probability. I.e., a maximal discrimination
of non \emph{typical} words will not significantly change the discriminated
symbol probabilities $P_{C_{i}}^{\otimes}(x|\mathbf{m})$. This indicates
that a random choice of the $\mathbf{w}^{(l)}$ for $l=1,2$ will
typically lead to similar beliefs $P_{C_{i}}^{\otimes}(x|\mathbf{m})$
as under $\mathbf{w}^{(1)}=\mathbf{w}^{(2)}=\mathbf{0}$.\\
Conversely it holds that if one code word at a small distance is maximally
discriminated then its probability typically dominates the probabilities
of the other terms in~(\ref{eq:otimesxMexp2u0}). 
\end{anm}
\begin{exa}
\label{exa:locMaxHard}We continue the example above. The discriminator
$\mathbf{m=}(\mathbf{r},\hat{\mathbf{c}}^{(a)},\mathbf{0})$ maximally
discriminates the distinguished word $\hat{\mathbf{c}}^{(a)}$ at
\[
\mathbf{u}=\mathbf{u}(\hat{\mathbf{c}}^{(a)})=(n-2d_{H}(\mathbf{r},\hat{\mathbf{c}}^{(a)}),n,0).\]
The discriminator complexity $\mathbb{U}$ is maximally ${(n+1)}^{2}$
as only this many different values of \[
\mathbf{u}=(n-2d_{H}(\mathbf{r},\mathbf{c}),n-2d_{H}(\hat{\mathbf{c}}^{(a)},\mathbf{c}),0)\]
exist. The complexity is then given by the computation of maximally
$(n+1)^{2}$ elements. As this has to be done $n$ times in the trellis
(see Appendix~\ref{sub:Trellis-Based-Algorithms}) the asymptotic
complexity becomes $O(n^{3})$ (for fixed trellis state complexity).
The computation will give by Theorem~\ref{thm:Asymptyp} that \[
P_{C_{i}}^{\otimes}(x|\mathbf{m})\asymp P_{C_{i}}^{(a)}(x|\mathbf{r})\mbox{ with }\hat{c}_{i}^{(a)}=\mbox{sign}(L_{i}^{\otimes}(\mathbf{m}))\]
as $\hat{\mathbf{c}}^{(a)}$ is distinguished and as all other words
can be assumed to be chosen independently.\\
I.e., $P^{\otimes}(\mathbf{u}|\mathbf{m})$ exhibits a peak of height
$1$ and the $P_{C_{i}}^{\otimes}(x|\mathbf{m})$ give the asymptotically
correct symbol probabilities. 
\end{exa}

\subsection{Well Defined Discriminators}

Example~\ref{exa:locMaxHard} shows that for the distinguished event
$\hat{\mathbf{c}}^{(a)}$ the \emph{hard} \emph{decision} discriminator
\[
\mathbf{m=}(\mathbf{r},\hat{\mathbf{c}}^{(a)},\mathbf{0})\mbox{ with }\hat{c}_{i}^{(a)}=\mbox{sign}(L_{i}^{\otimes}(\mathbf{m}))\]
 produces discriminated symbol beliefs close to the overall symbol
probabilities. The discriminator complexity $|\mathbb{U}|\leq{(n+1)}^{2}$
is thus sufficient to obtain the asymptotically correct decoding decision. 

\begin{anm}
\emph{(Equivalent Hard Decision Discriminators)} By (\ref{eq:otimesxMexp2u0})
the hard decision discriminators \[
\mathbf{m=}(\mathbf{r},\mathbf{w},\mathbf{0})\mbox{, }\mathbf{m}=(\mathbf{r},\mathbf{0},\mathbf{w})\mbox{, and }\mathbf{m}=(\mathbf{r},\mathbf{w},\mathbf{w})\]
are equivalent: For the three cases the same $\mathbb{C}_{i}^{(l)}(x|\mathbf{m})$
and $\mathbb{S}_{i}^{(l)}(x|\mathbf{m})$ and thus $P_{C_{i}}^{\otimes}(x|\mathbf{m})$
follow. In the sequel of this section we will (for symmetry reasons)
only consider the discriminators $\mathbf{m=}(\mathbf{r},\mathbf{w},\mathbf{w})$. 
\end{anm}
The discussion above shows that a discriminator with randomly chosen
$\mathbf{w}$ should give almost the same $L_{i}^{\otimes}(\mathbf{m})$
as $L_{i}^{\otimes}(\mathbf{r},\mathbf{0},\mathbf{0})$. If, however,
the discriminator is strongly dependent on the distinguished solution,
i.e., $\mathbf{w}=\hat{\mathbf{c}}^{(a)}$ then the correct solution
is found via $L_{i}^{\otimes}(\mathbf{m})$. This gives the following
definition and lemma.

\begin{Def}
\label{Def:(Well-Defined-Discriminator)}(Well Defined Discriminator)
A well defined discriminator $\mathbf{m}=(\mathbf{r},\mathbf{w},\mathbf{w})$
fulfils\emph{\begin{equation}
w_{i}=\mbox{sign}(L_{i}^{\otimes}(\mathbf{m}))\mbox{ for all }i.\label{eq:hard_dec_wohldefined}\end{equation}
} 
\end{Def}
\begin{lem}
\label{lem:Exis_approx_well_def}For a BSC and distinguished $\hat{\mathbf{c}}^{(a)}$
exists a well defined discriminator $\mathbf{m}=(\mathbf{r},\mathbf{w},\mathbf{w})$
with $w_{i},r_{i}\in\mathbb{B}$ such that $\hat{c}_{i}^{(a)}=w_{i}.$
\end{lem}
\begin{proof}
Set $\mathbf{m}=(\mathbf{r},\hat{\mathbf{c}}^{(a)},\hat{\mathbf{c}}^{(a)}).$
For this choice holds $\hat{\mathbf{c}}^{(a)}\in\mathbb{D}(\mathbf{m})$
and thus with Theorem~\ref{thm:Asymptyp} asymptotic equality. Moreover,
holds for a distinguished element that \[
P_{C_{i}}^{\otimes}(\hat{c}_{i}^{(a)}|\mathbf{m})\asymp P_{C_{i}}^{(a)}(\hat{c}_{i}^{(a)}|\mathbf{r})\asymp1\]
 and thus $\hat{c}_{i}^{(a)}=\mbox{sign}(L_{i}^{(a)}(\mathbf{r}))=\mbox{sign}(L_{i}^{\otimes}(\mathbf{m}))$.
\end{proof}
The definition of a well defined discriminator (\ref{eq:hard_dec_wohldefined})
can be used as an iteration rule, which gives Algorithm~\ref{alg:Iterative-Hard-Discrimination}.
The iteration thereby exhibits by Lemma~\ref{lem:Exis_approx_well_def}
a fixed point, which provably represents the distinguished solution.
Note that the employment of $\mathbf{w}^{(1)}=\mathbf{w}^{(2)}=\mathbf{w}$
is here handy as by $L_{i}^{\otimes}(\mathbf{m})$ only \emph{one}
common\emph{ }belief is available. This is contrast to Algorithm~\ref{alg:Loopy-Belief-Propagation}
where the employment of the two constituent beliefs generally give
that $\mathbf{w}^{(1)}\ne\mathbf{w}^{(2)}$.

\begin{algorithm}[H]
\begin{enumerate}
\item Set $\mathbf{m}=(\mathbf{r},\mathbf{0},\mathbf{0})$ and $\mathbf{w}=\mathbf{0}$.
\item Set $\mathbf{v}=\mathbf{w}$ and $w_{i}\leftarrow\mbox{sign}\,(L_{i}^{\otimes}(\mathbf{m}))$
for all $i$.
\item If $\mathbf{v}\neq\mathbf{w}$ then $\mathbf{m}=(\mathbf{r},\mathbf{w},\mathbf{w})$
and go to 2.
\item Set $\hat{\mathbf{c}}=\mathbf{w}.$\caption{\label{alg:Iterative-Hard-Discrimination}Iterative Hard Decision
Discrimination}

\end{enumerate}

\end{algorithm}

To understand the overall properties of the algorithm one needs to
consider its convergence properties and the existence of other fixed
points. A first intuitive assessment of the algorithm is as follows.
The decisions taken by $w_{i}=\mbox{sign}(L_{i}^{\otimes}(\mathbf{r},\mathbf{0},\mathbf{0}))$
should by (\ref{eq:logotimesextr}) lead to a smaller symbol error
probability than the one over $\mathbf{r}$. Overall these decisions
are based on $P^{\otimes}(\mathbf{u}|\mathbf{r},\mathbf{0},\mathbf{0})$.
This distribution is necessarily large in the vicinity of $\hat{u}_{0}=n-2\hat{t}$
with $\hat{t}$ the expected number of errors. 

The subsequent discrimination with $\mathbf{w}$ \emph{and} $\mathbf{r}$
will consider the vicinity of $\mathbf{c}$ more precisely if $\mathbf{w}\mathbf{c}^{T}$
is larger than $\mathbf{r}\mathbf{c}^{T}$: In this vicinity less
words exist, which gives that the $|\mathbb{S}(\mathbf{u}|\mathbf{m})|$
are smaller there. Smaller error probability in $\mathbf{w}$ is thus
with (\ref{eq:axelsmeckerei}) \emph{typically} equivalent to a better
discrimination in the vicinity of $\hat{\mathbf{c}}^{(a)}$. This
indicates that the discriminator $(\mathbf{r},\mathbf{w},\mathbf{w})$
is better than $(\mathbf{r},\mathbf{0},\mathbf{0})$. Hence, the new
$w_{i}\leftarrow\mbox{sign}(L_{i}^{\otimes}(\mathbf{r},\mathbf{w},\mathbf{w}))$
should exhibit again smaller error probability and so forth. If the
iteration ends then a stable solution is found. Finally, the solution
$\mathbf{w}=\hat{\mathbf{c}}^{(a)}$ is stable. %
\begin{figure}[tbh]
\noindent \begin{centering}
\subfigure[Initialisation]{\includegraphics[width=0.3\textwidth]{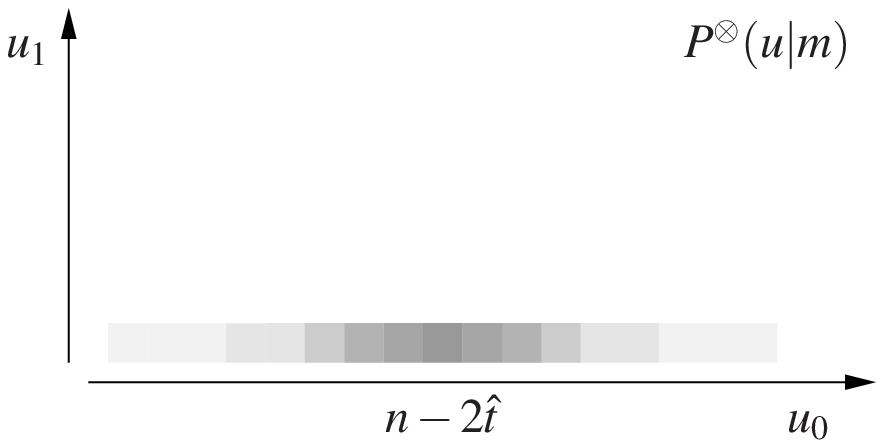}}~~~~\subfigure[Intermediate Step]{\includegraphics[width=0.3\textwidth]{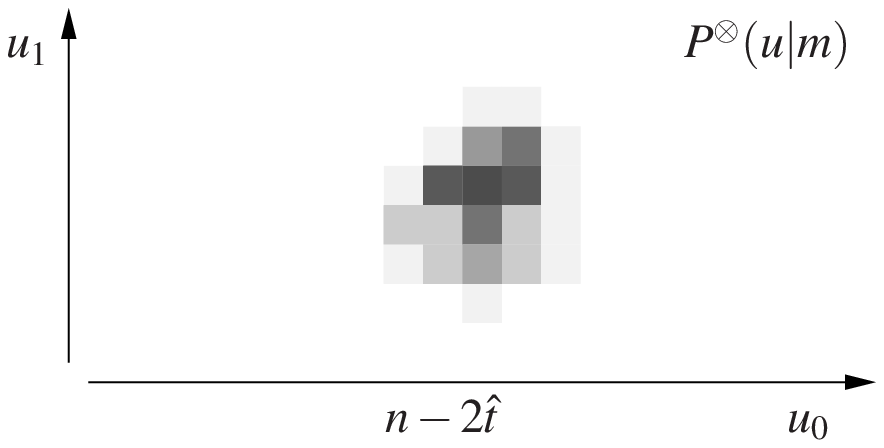}}~~~~\subfigure[Stable]{\includegraphics[width=0.3\textwidth]{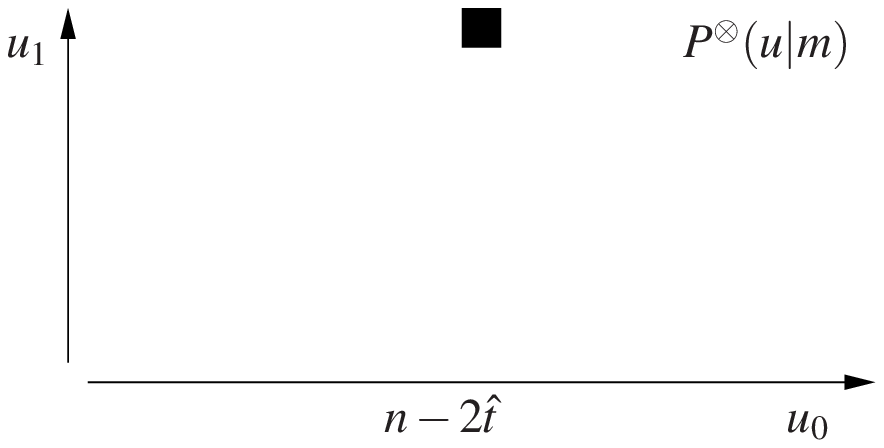}}
\par\end{centering}

\caption{\label{fig:Hard-Decision-Discrimination}Hard Decision Discrimination}

\end{figure}
 This behaviour is exemplary depicted in Figure~\ref{fig:Hard-Decision-Discrimination}
where the density of the squares represent the probability $P^{\otimes}(\mathbf{u}|\mathbf{m})$
of $\mathbf{u}=(u_{0},u_{1})$.

\subsection{Cross Entropy}

To obtain a \emph{quantitative} assessment of Algorithm~\ref{alg:Iterative-Hard-Discrimination}
we  use the following definition.

\begin{Def}
(Cross Entropy) The cross entropy \[
H(\mathbf{C}|\mathbf{w}\Vert\mathbf{r}):=\mathrm{E}_{\mathbf{C}}[H(\mathbf{s}|\mathbf{w})|\mathbf{r}]=-\sum_{\mathbf{s}\in\mathbb{C}}P_{\mathbf{C}}(\mathbf{s}|\mathbf{r})\log_{2}(P(\mathbf{s}|\mathbf{w}))\]
is the expectation of the uncertainty $H(\mathbf{s}|\mathbf{w})=-\log_{2}P(\mathbf{s}|\mathbf{w})$
under $\mathbf{r}$ and $\mathbf{c}\in\mathbb{E}(\mathbf{C})$.
\end{Def}
The cross entropy measures as the \noun{Kullback-Leibler} Distance
\[
D(\mathbf{C}|\mathbf{w}\Vert\mathbf{r}):=H(\mathbf{C}|\mathbf{w}\Vert\mathbf{r})-H(\mathbf{C}|\mathbf{r})\]
with \[
H(\mathbf{C}|\mathbf{r}):=\mathrm{E}_{\mathbf{C}}[H_{\mathbf{C}}(\mathbf{s}|\mathbf{r})|\mathbf{r}]=-\sum_{\mathbf{s}\in\mathbb{C}}P_{\mathbf{C}}(\mathbf{s}|\mathbf{r})\log_{2}(P_{\mathbf{C}}(\mathbf{s}|\mathbf{r})).\]
the similarity between the distributions $P(\mathbf{c}|\mathbf{r})$
and $P(\mathbf{s}|\mathbf{w})$. By \noun{Jensen}'s inequality it
is easy to show~\cite{Massey_Inf} that $D(\mathbf{C}|\mathbf{w}\Vert\mathbf{r})\geq0$
and thus \[
H(\mathbf{C}|\mathbf{w}\Vert\mathbf{r})\geq H(\mathbf{C}|\mathbf{r})\geq0.\]

The entropy $H(\mathbf{C}|\mathbf{r})$ is an information theoretic
measure of the number of probable words in $\mathbb{E}(\mathbf{C})$
under $\mathbf{r}$. To better explain the cross entropy $H(\mathbf{C}|\mathbf{w}\Vert\mathbf{r})$
we shortly review some results regarding the entropy. 

The \emph{typical} set $\mathbb{A}_{n\varepsilon}(\mathbf{C}|\mathbf{r})$
is given by the typical region \[
\mathbb{A}_{n\varepsilon}(\mathbf{C}|\mathbf{r})=\{\mathbf{c}\in\mathbb{E}(\mathbf{C}):|H(\mathbf{c}|\mathbf{r})-H(\mathbf{C}|\mathbf{r})|\leq n\varepsilon\}\]
 of word uncertainties\[
H(\mathbf{c}|\mathbf{r})=-\log_{2}P(\mathbf{c}|\mathbf{r}).\]
This definition directly gives \[
1\geq\sum_{\mathbb{A}_{n\varepsilon}(\mathbf{C}|\mathbf{r})}P(\mathbf{c}|\mathbf{r})=\sum_{\mathbb{A}_{n\varepsilon}(\mathbf{C}|\mathbf{r})}\exp_{2}(-H(\mathbf{c}|\mathbf{r}))\geq\exp_{2}(-H(\mathbf{C}|\mathbf{r})-n\varepsilon)\sum_{\mathbb{A}_{n\varepsilon}(\mathbf{C}|\mathbf{r})}1,\]
respectively, \[
P_{\mathbf{C}}(\mathbb{A}_{n\varepsilon}(\mathbf{C}|\mathbf{r})|\mathbf{r})=\sum_{\mathbb{A}_{n\varepsilon}(\mathbf{C}|\mathbf{r})}P(\mathbf{c}|\mathbf{r})=\sum_{\mathbb{A}_{n\varepsilon}(\mathbf{C}|\mathbf{r})}\exp_{2}(-H(\mathbf{c}|\mathbf{r}))\leq\exp_{2}(-H(\mathbf{C}|\mathbf{r})+n\varepsilon)\sum_{\mathbb{A}_{n\varepsilon}(\mathbf{C}|\mathbf{r})}1.\]
With \[
\sum_{\mathbb{A}_{n\varepsilon}(\mathbf{C}|\mathbf{r})}1=|\mathbb{A}_{n\varepsilon}(\mathbf{C}|\mathbf{r})|\]
 this leads to the bounds on the logarithmic set sizes\[
H(\mathbf{C}|\mathbf{r})+n\varepsilon\geq\log_{2}|\mathbb{A}_{n\varepsilon}(\mathbf{C}|\mathbf{r})|\geq H(\mathbf{C}|\mathbf{r})+\log_{2}(P_{\mathbf{C}}(\mathbb{A}_{n\varepsilon}(\mathbf{C}|\mathbf{r})|\mathbf{r}))-n\varepsilon\]
by the entropy. For many independent events in $\mathbf{r}$ the law
of large numbers gives for $\varepsilon>0$ that \[
P_{\mathbf{C}}(\mathbb{A}_{n\varepsilon}(\mathbf{C}|\mathbf{r})|\mathbf{r})\approx1\mbox{ and thus }H(\mathbf{C}|\mathbf{r})\approx\log_{2}(|\mathbb{A}_{n\varepsilon}(\mathbf{C}|\mathbf{r})|).\]

We investigate if a similar statement can be done for the cross entropy.
To do so first a \emph{cross} \emph{typical} \emph{set} $\mathbb{A}_{n\varepsilon}(\mathbf{C}|\mathbf{w}\Vert\mathbf{r})$
is defined by the region of typical word uncertainties: \begin{equation}
\min(H(\mathbf{S}|\mathbf{w}),H(\mathbf{C}|\mathbf{w}\Vert\mathbf{r}))-n\varepsilon\leq H(\mathbf{s}|\mathbf{w})\leq\max(H(\mathbf{S}|\mathbf{w}),H(\mathbf{C}|\mathbf{w}\Vert\mathbf{r}))+n\varepsilon.\label{eq:Def_cross_typical_set}\end{equation}
I.e., the region spans the typical set in $\mathbf{w}$ but includes
more words if $H(\mathbf{S}|\mathbf{w})\neq H(\mathbf{C}|\mathbf{w}\Vert\mathbf{r})$.
As the typical set in $\mathbf{w}$ is included this gives for large
$n$ that \[
P_{\mathbf{S}}(\mathbb{A}_{n\varepsilon}(\mathbf{C}|\mathbf{w}\Vert\mathbf{r})|\mathbf{w})\approx1\]
 and then in the same way as above the bounds on the logarithmic set
size\[
\max(H(\mathbf{S}|\mathbf{w}),H(\mathbf{C}|\mathbf{w}\Vert\mathbf{r}))+n\varepsilon\geq\log_{2}|\mathbb{A}_{n\varepsilon}(\mathbf{C}|\mathbf{w}\Vert\mathbf{r})|\geq\min(H(\mathbf{S}|\mathbf{w}),H(\mathbf{C}|\mathbf{w}\Vert\mathbf{r}))-n\varepsilon.\]
Moreover, holds by the definition of the cross entropy and the law
of large numbers that typically \[
P_{\mathbf{C}}(\mathbb{A}_{n\varepsilon}(\mathbf{C}|\mathbf{w}\Vert\mathbf{r})|\mathbf{r})\approx1\]
 is true, too. This gives that the cross typical set includes the
typical sets $\mathbb{A}_{n\varepsilon}(\mathbf{S}|\mathbf{w})$ and
$\mathbb{A}_{n\varepsilon}(\mathbf{C}|\mathbf{r})$, i.e., \begin{equation}
\mathbb{A}_{n\varepsilon}(\mathbf{C}|\mathbf{w}\Vert\mathbf{r})\supseteq\mathbb{A}_{n\varepsilon}(\mathbf{S}|\mathbf{w})\mbox{ and }\mathbb{A}_{n\varepsilon}(\mathbf{C}|\mathbf{w}\Vert\mathbf{r}))\supseteq\mathbb{A}_{n\varepsilon}(\mathbf{C}|\mathbf{r}).\label{eq:SuperTransferredSet}\end{equation}

If one wants to define a transfer vector $\mathbf{w}$ based on $\mathbf{r}$
one is thus interested to obtain a representation in $\mathbf{w}$
such that the logarithmic set size \[
\log_{2}|\mathbb{A}_{n\varepsilon}(\mathbf{C}|\mathbf{w}\Vert\mathbf{r})|\leq\max(H(\mathbf{S}|\mathbf{w}),H(\mathbf{C}|\mathbf{w}\Vert\mathbf{r}))+n\varepsilon\]
 is as small as possible. 

In the sequel we consider $P(\mathbf{s}|\mathbf{w})\propto P(\mathbf{w}|\mathbf{s})$
defined by (\ref{eq:defwahrscheinlich}). This probability is given
by \begin{equation}
P(\mathbf{s}|\mathbf{w})=\frac{\prod_{i=1}^{n}P(w_{i}|s_{i})}{\sum_{\mathbf{s}\in\mathbb{S}}P(\mathbf{w}|\mathbf{s})}=\prod_{i=1}^{n}\frac{P(s_{i}|w_{i})}{P(+1|w_{i})+P(-1|w_{i})}=\prod_{i=1}^{n}\frac{2^{s_{i}w_{i}}}{2^{w_{i}}+2^{-w_{i}}}.\label{eq:P(s|w)}\end{equation}
The cross entropy thus becomes\begin{align}
H(\mathbf{C}|\mathbf{w}\Vert\mathbf{r}) & =\sum_{\mathbf{s}\in\mathbb{C}}P_{\mathbf{C}}(\mathbf{s}|\mathbf{r})\sum_{i=1}^{n}(\log_{2}(2^{w_{i}}+2^{-w_{i}})-s_{i}w_{i})\nonumber \\
 & =\sum_{i=1}^{n}\log_{2}(2^{w_{i}}+2^{-w_{i}})-\sum_{i=1}^{n}\sum_{\mathbf{s}\in\mathbb{C}}P_{\mathbf{C}}(\mathbf{s}|\mathbf{r})s_{i}w_{i}\label{eq:CrossIndependent}\end{align}
and \[
\sum_{i=1}^{n}\sum_{\mathbf{s}\in\mathbb{C}}P_{\mathbf{C}}(\mathbf{s}|\mathbf{r})s_{i}w_{i}=\sum_{i=1}^{n}\mathrm{E}_{\mathbf{C},C_{i}}[x|\mathbf{r}]w_{i}=\sum_{i=1}^{n}w_{i}(P_{C_{i}}^{(c)}(+1|\mathbf{r})-P_{C_{i}}^{(c)}(-1|\mathbf{r})).\]
 This definition almost directly defines an optimal transfer.

\begin{lem}
\label{lem:Kullback}Equal logarithmic symbol probability ratios \[
w_{i}=L_{i}(\mathbf{w})=L_{i}^{(c)}(\mathbf{r})=\frac{1}{2}\log_{2}\frac{P_{C_{i}}^{(c)}(+1|\mathbf{r})}{P_{C_{i}}^{(c)}(-1|\mathbf{r})}\]
 and $P(\mathbf{s}|\mathbf{w})\propto P(\mathbf{w}|\mathbf{s})$ defined
by (\ref{eq:defwahrscheinlich}) minimise cross entropy $H(\mathbf{C}|\mathbf{w}\Vert\mathbf{r})$
and \noun{Kullback}-\noun{Leibler} distance $D(\mathbf{C}|\mathbf{w}\Vert\mathbf{r})$. 
\end{lem}
\begin{proof}
First it holds by (\ref{eq:defwahrscheinlich}) that \[
w_{i}=L_{i}(\mathbf{w})=\frac{1}{2}\log_{2}\frac{\exp_{2}(+w_{i})}{\exp_{2}(-w_{i})}.\]
A differentiation of (\ref{eq:CrossIndependent}) leads to \[
\frac{\partial}{\partial w_{i}}H(\mathbf{C}|\mathbf{w}\Vert\mathbf{r})=\sum_{\mathbf{s}\in\mathbb{C}}P_{\mathbf{C}}(\mathbf{s}|\mathbf{r})(\tanh_{2}(w_{i})-x_{i})=\tanh_{2}(w_{i})-\sum_{\mathbf{s}\in\mathbb{C}}c_{i}P_{\mathbf{C}}(\mathbf{s}|\mathbf{r})\stackrel{!}{=}0\]
with $\tanh_{2}(x)=(2^{x}-2^{-x})/(2^{x}+2^{-x})$. This directly
gives that \[
\tanh_{2}(w_{i})=\sum_{\mathbf{s}\in\mathbb{C}}c_{i}P_{\mathbf{C}}(\mathbf{s}|\mathbf{r})=P_{C_{i}}^{(c)}(+1|\mathbf{r})-P_{C_{i}}^{(c)}(-1|\mathbf{r}).\]
As \[
\tanh_{2}(L_{i}^{(c)}(\mathbf{r}))=P_{C_{i}}^{(c)}(+1|\mathbf{r})-P_{C_{i}}^{(c)}(-1|\mathbf{r})\]
and $\frac{\partial}{\partial w_{i}}H(\mathbf{C}|\mathbf{r})=0$ this
is equivalent to the statement of the lemma.
\end{proof}
I.e., the definition of $\mathbf{w}$ by $L_{i}(\mathbf{w})=L_{i}^{(c)}(\mathbf{r})$
is a consequence of the independence assumption (\ref{eq:defwahrscheinlich}).
Especially interesting is that $L_{i}(\mathbf{w})=L_{i}^{(c)}(\mathbf{r})$
directly implies that\[
H(\mathbf{S}|\mathbf{w})=H(\mathbf{C}|\mathbf{w}\Vert\mathbf{r}),\]
 which gives with (\ref{eq:SuperTransferredSet}) that\[
\mathbb{A}_{n\varepsilon}(\mathbf{C}|\mathbf{w}\Vert\mathbf{r})=\mathbb{A}_{n\varepsilon}(\mathbf{S}|\mathbf{w})\mbox{ and thus }\mathbb{A}_{n\varepsilon}(\mathbf{S}|\mathbf{w})\supseteq\mathbb{A}_{n\varepsilon}(\mathbf{C}|\mathbf{r}).\]
A belief representing transfer vector $\mathbf{w}$ thus typically
describes all probable codewords.\\
By reconsidering the definition of the cross typical set in (\ref{eq:Def_cross_typical_set})
the in $\mathbf{r}$ and $\mathbf{C}$ typical set $\mathbb{A}_{n\varepsilon}(\mathbf{C}|\mathbf{r})$
is (in the mean) contained in the set of in $\mathbf{w}$ probable
words $\mathbf{s}\in\mathbb{S}$ if \[
H(\mathbf{S}|\mathbf{w})\geq H(\mathbf{C}|\mathbf{w}\Vert\mathbf{r}).\]
Hereby the set of probable words is defined by only considering the
right hand side inequality of (\ref{eq:Def_cross_typical_set}).

\subsection{Discriminator Entropy}

In this section the considerations are extended to the discrimination.
To do so we use in equivalence to (\ref{eq:CrossIndependent}) the
following definition. 

\begin{Def}
\label{Def:(Discriminated-Cross-Entropy)}(Discriminated Cross Entropy)
The discriminated cross entropy is \begin{align*}
H(\mathbf{C}^{\otimes}|\mathbf{w}\Vert\mathbf{m}): & =-\sum_{\mathbf{u}\in\mathbb{U}}P^{\otimes}(\mathbf{u}|\mathbf{m})\log_{2}P(\mathbf{s}|\mathbf{w}):=\sum_{i=1}^{n}\log_{2}(2^{w_{i}}+2^{-w_{i}})-\sum_{\mathbf{u}\in\mathbb{U}}w_{i}s_{i}\cdot P_{C_{i}}^{\otimes}(s_{i},\mathbf{u}|\mathbf{m})\\
 & =\sum_{i=1}^{n}\log_{2}(2^{w_{i}}+2^{-w_{i}})-w_{i}\mathrm{E}_{\mathbf{U}}^{\otimes}[c_{i}|\mathbf{m}]\end{align*}
with (\ref{eq:P(s|w)}) and $\mathrm{E}_{\mathbf{U}}^{\otimes}[c_{i}|\mathbf{m}]=P_{C_{i}}^{\otimes}(+1|\mathbf{m})-P_{C_{i}}^{\otimes}(-1|\mathbf{m})$.
\end{Def}
Note that this definition again uses the correspondence of $\mathbf{u}$
and $\mathbf{s}$. Even though by a discrimination not all words are
independently considered, a word uncertainty consideration is still
possible by attributing appropriate probabilities. Lemma~\ref{lem:Kullback}
directly gives that the discriminated cross entropy is always larger
than or equal to the \emph{discriminated symbol entropy} $H(\mathbf{C}^{\otimes}\Vert\mathbf{m})$,
i.e., \[
H(\mathbf{C}^{\otimes}|\mathbf{w}\Vert\mathbf{m})\geq-\sum_{\mathbf{u}\in\mathbb{U}}P^{\otimes}(\mathbf{u}|\mathbf{m})\log_{2}P(\mathbf{s}|\mathbf{L}^{\otimes}(\mathbf{m}))=:H(\mathbf{C}^{\otimes}\Vert\mathbf{m})\]
 The discriminator entropy measures the uncertainty of the discriminated
decoding decision, i.e., the number of words in $\mathbb{S}$ that
need to be considered. This directly gives the following theorem.

\begin{thm}
\label{thm:Hard-decoding=00003Dminimal_DH}The decoding problem for
a distinguished word is equivalent to the solution of\emph{ \begin{equation}
w_{i}=\mbox{sign}(L_{i}^{\otimes}(\mathbf{m}))\label{eq:extrinsic Definition}\end{equation}
}with the discriminated symbol entropy $H(\mathbf{C}^{\otimes}\Vert\mathbf{m})<1$
and $\mathbf{m}=(\mathbf{r},\mathbf{w},\mathbf{w})$.
\end{thm}
\begin{proof}
For $w_{i}=\hat{c}_{i}^{(a)}$ is $\hat{\mathbf{c}}^{(a)}\in\mathbb{D}(\mathbf{m})$.
This gives with Lemma~\ref{lem:Exis_approx_well_def} for the discriminated
distribution that \[
P^{\otimes}(\mathbf{u}|\mathbf{m})\asymp P^{(a)}(\mathbf{u}|\mathbf{r}).\]
 As $\hat{\mathbf{c}}^{(a)}$ is a distinguished solution this gives
$P^{(a)}(\mathbf{u}(\hat{\mathbf{c}}^{(a)})|\mathbf{r})\approx1$
or equivalently $H(\mathbf{C}^{\otimes}\Vert\mathbf{m})\approx0.$ 

The discriminated symbol entropy $H(\mathbf{C}^{\otimes}\Vert\mathbf{m})$
estimates by $\exp_{2}H(\mathbf{C}^{\otimes}\Vert\mathbf{m})$ the
logarithmic number of elements in the set of probable words in $\mathbb{S}$.
Any solution $\mathbf{m}$ with \[
H(\mathbf{C}^{\otimes}\Vert\mathbf{m})<1\]
thus exhibits \emph{one} word $\mathbf{s}$ with $P^{\otimes}(\hat{\mathbf{u}}|\mathbf{m})\approx1$.
I.e., one has a discriminated distribution $P^{\otimes}(\mathbf{u}|\mathbf{m})$
that contains just one peak of height almost one at $\hat{\mathbf{u}}$.
As only one word distributes the decisions by (\ref{eq:extrinsic Definition})
give this word, or equivalently that $\hat{\mathbf{u}}=\mathbf{u}(\mathbf{w})$.
Hence, $\hat{\mathbf{c}}=\mathbf{w}$ is maximally discriminated. 

This directly implies that the obtained $\hat{\mathbf{c}}$ needs
to be a codeword of the coupled code: Both distributions $P^{(l)}(\mathbf{u}|\mathbf{m})$
are used for the single word description $P^{\otimes}(\mathbf{u}|\mathbf{m})\neq0$.
Hence, both codes contain the in $\mathbf{u}$ maximally discriminated
word $\hat{\mathbf{c}}$, which gives (by the definition of the dually
coupled code) that this word is an overall codeword.

Assume that $\hat{\mathbf{c}}\neq\hat{\mathbf{c}}^{(a)}$ represents
a non distinguished word. With Remark~\ref{anm:(Distance)} this
word needs to exhibit a large distance to $\mathbf{r}$. Typically
many words $\mathbf{c}\in\mathbb{C}^{(a)}$ exist at such a large
distance. By (\ref{eq:asympgle1}) these words are considered in the
computation of $P^{\otimes}(\mathbf{u}|\mathbf{m})$. Thus $P^{\otimes}(\mathbf{u}|\mathbf{m})$
is not in the form of a peek, which gives that $H(\mathbf{C}^{\otimes}\Vert\mathbf{m})>1$.
As this is a contradiction no other solution of (\ref{eq:extrinsic Definition})\textbf{
$\mathbf{w}$ }but $\mathbf{w}=\hat{\mathbf{c}}^{(a)}$ may exhibit
a discriminated symbol entropy $H(\mathbf{C}^{\otimes}\Vert\mathbf{m})<1$.
\end{proof}
\begin{anm}
(\emph{Typical} \emph{Decoding}) The proof of the theorem indicates
that any code word $\mathbf{c}\in\mathbb{C}^{(a)}$ with small distance
to $\mathbf{r}$ may give rise to a well defined discriminator $\mathbf{m}$
with $H(\mathbf{C}^{\otimes}\Vert\mathbf{m})<1$ and $\mathbf{w}=\mathbf{c}$.
Hence, a low entropy solution of the equation is not equivalent to
ML decoding. However, if the code rate is below capacity and a long
code is employed only one distinguished word exists. 
\end{anm}
Theorem~\ref{thm:Hard-decoding=00003Dminimal_DH} gives that Algorithm~\ref{alg:Iterative-Hard-Discrimination}
fails in finding the distinguished word if either the stopping criterion
is never fulfilled (it runs infinitely long) or the solution exhibits
a large discriminated symbol entropy. To investigate these cases consider
the following Lemma.

\begin{lem}
\label{lem:Hard_Discriminated_Cross_Entropy}It holds that \emph{\begin{equation}
w_{i}\leftarrow\mbox{sign}(L_{i}^{\otimes}(\mathbf{m}))\mbox{ for all }i\label{eq:Update-hard-discriminator}\end{equation}
}minimises the cross entropy $H(\mathbf{C}^{\otimes}|\mathbf{w}\Vert\mathbf{m})$
under the constraint $w_{i}\in\mathbb{B}$. 
\end{lem}
\begin{proof}
The cross entropy $H(\mathbf{C}^{\otimes}|\mathbf{w}\Vert\mathbf{m})$
is given by \[
H(\mathbf{C}^{\otimes}|\mathbf{w}\Vert\mathbf{m})=\sum_{i=1}^{n}\log_{2}(2^{w_{i}}+2^{-w_{i}})-w_{i}\cdot\tanh(L_{i}^{\otimes}(\mathbf{m})).\]
 The cross entropy is under constant $|w_{i}|$ or $w_{i}\in\mathbb{B}$
obviously minimal for \[
\mbox{sign (}w_{i}\cdot\tanh(L_{i}^{\otimes}(\mathbf{m})))=1,\]
 which is the statement of the Lemma.
\end{proof}
The algorithm fails if the iteration does not converge. However, the
lemma gives that (\ref{eq:Update-hard-discriminator}) minimises in
each step of the iteration the cross entropy towards $\mathbf{w}$.
This is equivalent to \[
H(\mathbf{C}^{\otimes}|\mathbf{m}\Vert\mathbf{m})\geq H(\mathbf{C}^{\otimes}|\mathbf{w}\Vert\mathbf{m}).\]
This cross entropy is with $H(\mathbf{C}^{\otimes}|\mathbf{w}\Vert\mathbf{m})\geq\min_{\mathbf{v}}H(\mathbf{C}^{\otimes}|\mathbf{v}\Vert\mathbf{m})=H(\mathbf{C}^{\otimes}\Vert\mathbf{m})$
always larger than the overall discriminated symbol entropy. Furthermore,
holds by the optimisation rule that \[
H(\mathbf{S}|\mathbf{w})\geq H(\mathbf{C}^{\otimes}|\mathbf{w}\Vert\mathbf{m})\geq H(\mathbf{C}^{\otimes}\Vert\mathbf{m}),\]
 which gives that the typical set under the discrimination remains
included. The subsequent step will therefore continue to consider
this set. If the discriminated cross entropy does not further decrease
one thus obtains the same $\mathbf{w}$, which is a fixed point. 

This observation is similar to the discussion above. A discriminator
$\mathbf{m}$ describes environments with words close to $\mathbf{r}$.
A minimisation of the cross entropy can be considered as an optimal
description of this environment under the independence assumption
(and the imposed hard decision constraint). If this knowledge is processed
iteratively then these environments should be better and better investigated.
The discriminated symbol entropy $H(\mathbf{C}^{\otimes}\Vert\mathbf{m})$
will thus typically decrease. For an infinite loop this is not fulfilled,
i.e., such a loop is unlikely or non typical.

Moreover, the iterative algorithm fails if a stable solution $w_{i}=\mbox{sign}(L_{i}^{\otimes}(\mathbf{m}))$
with $\mathbf{w}\neq\hat{\mathbf{c}}$ is found. These solutions exhibit
with the proof of Theorem~\ref{thm:Hard-decoding=00003Dminimal_DH}
large discriminated symbol entropy $H(\mathbf{C}^{\otimes}\Vert\mathbf{m})$
(many words are probable) and thus small $|L_{i}^{\otimes}(\mathbf{m})|$.
However, solutions with small $|L_{i}^{\otimes}(\mathbf{m})|$ seem
unlikely as these values are usually for $\mathbf{w}=\mathbf{0}$
already relatively large and Lemma~\ref{lem:Hard_Discriminated_Cross_Entropy}
indicates that these values will become larger in each step.

\begin{anm}
\label{anm:(Improvements-hard) }(\emph{Improvements}) If the algorithm
fails due to a well defined discriminator of large cross entropy then
an appropriately chosen increase of the discriminator complexity should
improve the algorithm. To increase the discrimination complexity under
hard decisions one may use discriminators $\mathbf{w}^{(1)}\neq\mathbf{w}^{(2)}$.
One possibility is hereby to reuse the old transfer vector by \[
\mathbf{w}^{(2)}=\mathbf{w}^{(1)}\mbox{ and }w_{i}^{(1)}=\mbox{sign}(L_{i}^{\otimes}(\mathbf{m}))\]
 in Step 2 of the iterative algorithm. The complexity of the algorithm
will then, however, increase to $O(n^{4})$. 

On the other hand the complexity can be strongly decreased without
loosing the possibility to maximally discriminate the distinguished
word. First holds that only (distinguished) words up to some distance
$t$ from the received word contribute to $\mathbf{L}^{\otimes}(\mathbf{m})$.
One may thus decrease (if full discrimination of the distance to $\mathbf{w}$
is used) the complexity $|\mathbb{U}|\leq t\cdot(n+1)$ if only those
values are computed. 

A further reduction is obtained by the use of erasures in $\mathbf{w}$,
i.e., by $w_{i}\in\{-1,0,+1\}$ and \[
w_{i}=\frac{\mbox{sign}(L_{i}^{\otimes}(\mathbf{m}))-r_{i}}{2}\]
 in Step~2 of the algorithm. Note that this discrimination has only
complexity $|\mathbb{U}|\leq t^{2}$ as $u_{1}(\hat{\mathbf{c}})\leq\mathbf{w}\mathbf{w}^{T}\leq t$
is typically fulfilled. 

It remains to show that the distinguished solution is stable. We do
this here with the informal proof: For $\hat{c}_{i}=\mbox{sign}(L_{i}^{\otimes}(\mathbf{m}))$
one obtains that $w_{i}=\hat{c}_{i}$ if $\hat{c}_{i}\neq r_{i}$
and 0 else. Hence one obtains that $u_{1}(\hat{\mathbf{c}})=\mathbf{r}\mathbf{\hat{c}}^{T}$
and $u_{1}(\hat{\mathbf{c}})=\mathbf{w}\mathbf{w}^{T}$. If only one
word $\mathbf{s}$ exists for these values $u_{l}$ then this discriminator
$(\mathbf{r},\mathbf{w},\mathbf{w})$ is surely maximally discriminating.
First holds that if $u_{1}(\hat{\mathbf{c}})=\mathbf{w}\mathbf{w}^{T}$
that then $s_{i}=w_{i}$ for is uniquely defined $w_{i}\neq0$. Under
$u_{1}(\hat{\mathbf{c}})$ then the other symbols are uniquely defined
to $s_{i}=r_{i}$ as $u_{1}(\hat{\mathbf{c}})=\mathbf{r}\hat{\mathbf{c}}^{T}$
is the unique maximum of $u_{1}(\mathbf{s})$ under $s_{i}=w_{i}$
for $w_{i}\neq0$. 

The overall complexity of this algorithm is thus smaller than $O(n\cdot t^{2})$
respectively $O(n\cdot t^{3})$ for a discrimination with $\mathbf{w}^{(1)}\neq\mathbf{w}^{(2)}$. 
\end{anm}

%% file: TR-gauss.tex
\section{Approximations }

The last section indicates that an iterative algorithm with discriminated
symbol probabilities should outperform the iterative propagation of
only the constituent beliefs. However, the discriminator approach
was restricted to problems with small discriminator complexity $|\mathbb{U}|$.
\\
In this form and Remark~\ref{anm:(Real-Valued-Discriminators)} the
algorithm does not apply for example to AWGN channels. In this section
discriminator based decoding is generalised to real valued $\mathbf{w}^{(l)}$
and $\mathbf{r}$, and hence generally $|\mathbb{U}|=|\mathbb{S}|$. 

For a prohibitively large discriminator complexity $|\mathbb{U}|$
the distributions $P_{C_{i}}^{(l)}(x,\mathbf{u}|\mathbf{m}^{(l)})$
can not be practically computed; only an \emph{approximation} is feasible.
This approximation is usually done via a probability density, i.e.,
\[
p_{C_{i}}^{(l)}(x,\mathbf{u}|\mathbf{m}^{(l)})\mathrm{d}\mathbf{u}\approx P_{C_{i}}^{(l)}(x,\mathbf{u}|\mathbf{m}^{(l)})\]
where $p_{C_{i}}^{(l)}(x,\mathbf{u}|\mathbf{m}^{(l)})$ is described
by a small number of parameters.

\begin{anm}
(\emph{Representation} \emph{and Approximation}) The use of an approximation
changes the premise compared to the last section. There we assumed
that the representation complexity of the discriminator is limited
but that the computation is perfect. In this section we assume that
the discriminator is generally globally maximal but that an approximation
is sufficient. 
\end{anm}
An estimation of a distribution may be performed by a \emph{histogram}
given by the rule \[
\mathbb{U}_{\varepsilon}(\mathbf{u}|\mathbf{m}):=\bigcup_{|\mathbf{v}-\mathbf{u}|<\epsilon}\mathbb{U}(\mathbf{v}|\mathbf{m})\]
and the quantisation $\varepsilon$. These values can be approximated
(see Appendix~\ref{sub:Trellis-Based-Algorithms}) with an algorithm
that exhibits a comparable complexity as the one for the computation
of the hard decision values. For a sufficiently small $\varepsilon$
one obviously obtains a sufficient approximation. Here the complexity
remains of the order $O(n^{3})$. It may, however, be reduced as in
Remark~\ref{anm:(Improvements-hard) }. 

\begin{anm}
\label{anm:(Information-Distance)}(\emph{Uncertainty and} \emph{Distance})
The approach with histograms is equivalent to assuming that words
with similar $\mathbf{u}$ do not need to be distinguished; a discrimination
of $\mathbf{s}^{(1)}$ and $\mathbf{s}^{(2)}$ is assumed to be not
necessary if the {}``uncertainty distance'' \[
\mathrm{d}_{H}(\mathbf{u}(\mathbf{s}^{(1)}),\mathbf{u}(\mathbf{s}^{(2)}))=\mathrm{d}_{H}(\mathbf{u}^{(1)},\mathbf{u}^{(2)})=\sum_{l=0}^{2}\Vert H(\mathbf{s}^{(1)}|\mathbf{w}^{(l)})-H(\mathbf{s}^{(2)}|\mathbf{w}^{(l)})\Vert=\sum_{l=0}^{2}\Vert u_{l}^{(1)}-u_{l}^{(2)}\Vert\]
of $\mathbf{s}^{(1)}$ and $\mathbf{s}^{(2)}$ is smaller than some
$\varepsilon$. The error that occurs by \begin{equation}
P_{C_{i}}^{\otimes}(x|\mathbf{m})\propto\intop_{\mathbb{U}}\frac{p_{C_{i}}^{(1)}(x,\mathbf{u}|\mathbf{m}^{(1)})p_{C_{i}}^{(2)}(x,\mathbf{u}|\mathbf{m}^{(2)})}{p_{C_{i}}(x,\mathbf{u}|\mathbf{m})}\mathrm{d}\mathbf{u}\label{eq:extrins_distri}\end{equation}
can for sufficiently small $\varepsilon$ usually be neglected. 

Note that another approach is to approximate only in $u_{0}$ and
continue to use a limited discriminator complexity in $\mathbf{w}$
(by for example hard decision $w_{i}\in\mathbb{B}$), which gives
an exact discrimination in $u_{1}.$ 
\end{anm}

\subsection{\noun{Gauss} Discriminators }

Distributions are usually represented via parameters defined by \emph{expectations}.
This is done as the law of large numbers shows that these expectations
can be computed out of a statistics. Given these values then the unknown
distributions may be approximated by maximum entropy~\cite{Jaynes03}
densities.

\begin{exa}
The simplest method to approximate distributions by probability densities
is to assume that no extra knowledge is available over $\mathbf{u}$.
This leads to the maximal entropy {}``distributions'' (In \noun{Bayes}'
estimation theory this is equivalent to a non proper prior) with stripped
$\mathbf{u}$ \[
P_{C_{i}}^{(l)}(x|\mathbf{m}^{(l)})\approx P_{C_{i}}^{(l)}(x,\mathbf{u}|\mathbf{m}^{(l)})\mbox{ and }P_{C_{i}}(x|\mathbf{m})\approx P_{C_{i}}(x,\mathbf{u}|\mathbf{m}),\]
which is equivalent to $L_{i}^{{\scriptscriptstyle \boxtimes}}(\mathbf{m})=0$
as then $P_{C_{i}}^{{\scriptscriptstyle \boxtimes}}(\mathbf{u}|x,\mathbf{m})=1$
or \[
L_{i}^{\otimes}(\mathbf{m})=r_{i}+\breve{L}_{i}^{(1)}(\mathbf{m}^{(1)})+\breve{L}_{i}^{(2)}(\mathbf{m}^{(2)})\]
and thus implicitly to Algorithm~\ref{alg:Loopy-Belief-Propagation}.
Note that the derived tools do not give rise to a further evaluation
of this approach: A discrimination in the sense defined above does
not take place.
\end{exa}
The additional expectations considered here are the mean values $\mu_{l}$
and the correlations $\phi_{l,k}$. These are for the given correlation
map \begin{equation}
u_{l}(\mathbf{s})=\sum_{i=1}^{n}w_{i}^{(l)}s_{i}\label{eq:Correlation}\end{equation}
\label{sub:Moments}and $|s_{i}|=1$ defined by \begin{align*}
\mu_{l}^{(h)} & =\mathrm{E}_{\mathbf{C}^{(h)}}[u_{l}|\mathbf{m}^{(h)}]=\sum_{i=1}^{n}\sum_{\mathbb{C}^{(h)}}w_{i}^{(l)}c_{i}P^{(h)}(\mathbf{c}|\mathbf{m}^{(h)})\\
 & =\sum_{i=1}^{n}w_{i}^{(l)}(P_{C_{i}}^{(h)}(+1|\mathbf{m}^{(h)})-P_{C_{i}}^{(h)}(+1|\mathbf{m}^{(h)}))=\sum_{i=1}^{n}w_{i}^{(l)}\mathrm{E}_{\mathbf{C}^{(h)}}[c_{i}|\mathbf{m}^{(h)}]\end{align*}
and \begin{align*}
\left[\phi_{l,k}^{(h)}\right]^{2}+\mu_{l,j}\mu_{l,k}=\mathrm{E}_{\mathbf{C}^{(h)}}[u_{l}u_{k}|\mathbf{m}^{(h)}] & =\sum_{i=1}^{n}\sum_{j=1}^{n}w_{i}^{(l)}w_{j}^{(k)}\sum_{{\mathbf{c}\in\mathbb{C}}^{(h)}}c_{i}c_{j}P^{(h)}(\mathbf{c}|\mathbf{m}^{(h)})\\
 & =\sum_{i=1}^{n}\sum_{j=1}^{n}w_{j}^{(l)}w_{i}^{(k)}\mathrm{E}_{\mathbf{C}^{(h)}}[c_{i}c_{j}|\mathbf{m}^{(h)}].\end{align*}
The complexity of the computation of each value $\mu_{l}^{(h)}$ and
$\phi_{l,k}^{(h)}$ is here (see Appendix~\ref{sub:Trellis-Based-Algorithms})
comparable to the complexity of the BCJR Algorithm, i.e., for fixed
trellis state complexity $O(n)$. 

For known mean values and variances the maximum entropy density is
the \noun{Gauss} density. This density is with the following Lemma
especially suited for discriminator based decoding.

\begin{lem}
\label{lem:gauss-gauss}For long codes with small trellis complexity
one obtains asymptotically a\noun{ Gauss} density for $P^{(l)}(\mathbf{u}|\mathbf{m}^{(l)})$
and $P(\mathbf{u}|\mathbf{m})$.
\end{lem}
\begin{proof}
The values $u_{l}(\mathbf{c})$ are obtained by the correlation given
in (\ref{eq:Correlation}). For $P(\mathbf{u}|\mathbf{m})$ this is
equivalent to a sum of independent random values. I.e., $P(\mathbf{u}|\mathbf{m})$
is by the central limit theorem \noun{Gauss} distributed. For long
codes with small trellis state complexity and many considered words
the same holds true for $P^{(h)}(\mathbf{u}|\mathbf{m}^{(h)})$. In
this case the limited code memory gives sufficiently many independent
regions of subsequent code symbols. I.e., the correlation again leads
to a sum of many independent random values. 
\end{proof}
\begin{anm}
(\emph{Notation}) The \noun{Gauss} approximated symbol probability
distributions are here denoted by a hat, i.e., $\hat{p}_{C_{i}}^{(l)}(x,\mathbf{u}|\mathbf{m})$
and $\hat{p}_{C_{i}}(x,\mathbf{u}|\mathbf{m})$. The same is done
for the approximated logarithmic symbol probability ratios.
\end{anm}
The constituent \noun{Gauss} approximations then imply the approximation
of $P^{\otimes}(\mathbf{u}|\mathbf{m})$ by\[
\hat{p}^{\otimes}(\mathbf{u}|\mathbf{m})\propto\frac{\hat{p}^{(2)}(\mathbf{u}|\mathbf{m}^{(2)})\hat{p}_{C_{i}}^{(1)}(\mathbf{u}|\mathbf{m}^{(1)})}{\hat{p}(\mathbf{u}|\mathbf{m})}\]
and thus approximated discriminated symbol probabilities (for the
computation see Appendix~\ref{sub:Appendix_a2}) \begin{equation}
\hat{P}_{C_{i}}^{\otimes}(x|\mathbf{m})\propto\intop_{\mathbb{U}}\frac{\hat{p}_{C_{i}}^{(1)}(x,\mathbf{u}|\mathbf{m}^{(1)})\hat{p}_{C_{i}}^{(2)}(x,\mathbf{u}|\mathbf{m}^{(2)})}{\hat{p}_{C_{i}}(x,\mathbf{u}|\mathbf{m})}\mathrm{d}\mathbf{u}.\label{eq:approximated discriminated symbol probabilities}\end{equation}

This approximation is obtained via other approximations. Its quality
can thus not be guaranteed as before. To use the approximated discriminated
symbol probabilities in an iteration one therefore first has to check
the validity of (\ref{eq:approximated discriminated symbol probabilities}). 

By some choice of $\mathbf{w}^{(1)}$ and $\mathbf{w}^{(2)}$ the
approximations of the constituent distribution are performed in an
environment $\mathbb{A}_{n\varepsilon}(\mathbf{C}^{(l)}|\mathbf{m}^{(l)})$
where $\hat{p}^{(l)}(\mathbf{u}|\mathbf{m}^{(l)})$ is large. The
overall considered region is given by $\mathbb{A}_{n\varepsilon}(\mathbf{S}|\mathbf{m})$
defined by $\hat{p}(\mathbf{u}|\mathbf{m})$. This overall region
represents the possible overall words. The approximation is surely
valid if the possible code words of the $l$-th constituent code under
$\mathbf{m}^{(l)}$ are included in this region. I.e.,  the conditions
\begin{equation}
\mathbb{A}_{n\varepsilon}(\mathbf{C}^{(l)}|\mathbf{m}^{(l)})\subseteq\mathbb{A}_{n\varepsilon}(\mathbf{S}|\mathbf{m})\label{eq:g_regions}\end{equation}
for $l=1,2$ have to be fulfilled. In this case the description of
the last section applies as then the approximation is typically good.

\begin{anm}
That this consideration is necessary becomes clear under the assumption
that the constituent \noun{Gauss} approximations do \emph{not} consider
the same environments. In this case their mean values strongly differ.
The approximation of the discriminated distribution, however, will
therefore consider regions with a large distance to the mean. The
obtained results are then not predictable as a \noun{Gauss} approximation
is only good for the words assumed to be probable, i.e., close to
its mean value. Under (\ref{eq:g_regions}) this can not happen.
\end{anm}
The condition (\ref{eq:g_regions}) is -- in respect to the set sizes
-- fulfilled if \[
H(\mathbf{S}|\mathbf{m})\geq H(\mathbf{C}^{(l)}|\mathbf{m}\Vert\mathbf{m}^{(l)})\]
as this is equivalent to \[
\mathbb{A}_{n\varepsilon}(\mathbf{S}|\mathbf{m})\supseteq\mathbb{A}_{n\varepsilon}(\mathbf{C}^{(l)}|\mathbf{m}\Vert\mathbf{m}^{(l)}),\]
which gives with (\ref{eq:SuperTransferredSet}) that \[
\mathbb{A}_{n\varepsilon}(\mathbf{S}|\mathbf{m})\supseteq(\mathbb{A}_{n\varepsilon}(\mathbf{C}^{(l)}|\mathbf{m}\Vert\mathbf{m}^{(l)})\cap\mathbb{C}^{(l)})\supseteq\mathbb{A}_{n\varepsilon}(\mathbf{C}^{(l)}|\mathbf{m}^{(l)}).\]
However, by (\ref{eq:g_regions}) not only the set sizes but also
the words need to match. With\[
H(\mathbf{C}^{(l)}|\mathbf{m}\Vert\mathbf{m}^{(l)})=\sum_{i=0}^{n}H(C_{i}^{(l)}|r_{i}+w_{i}^{(1)}+w_{i}^{(2)}\Vert\mathbf{m}^{(l)})\]
we therefore employ the symbol wise conditions\begin{equation}
H(C_{i}|r_{i}+w_{i}^{(1)}+w_{i}^{(2)})\geq H(C_{i}^{(l)}|r_{i}+w_{i}^{(1)}+w_{i}^{(2)}\Vert\mathbf{m}^{(l)}).\label{eq:SymbolConditionDisc}\end{equation}
As all symbols are independently considered the conditions (\ref{eq:g_regions})
are then typically fulfilled. 

A decoding decision is again found if $\hat{H}(\mathbf{C}^{\otimes}\Vert\mathbf{m})<1$
under (\ref{eq:SymbolConditionDisc}). To find such a solution we
propose to minimise in each step $\hat{H}(\mathbf{C}^{\otimes}|\mathbf{v}\Vert\mathbf{m})$
under the condition (\ref{eq:SymbolConditionDisc}) of code $l$. 

As then (\ref{eq:SymbolConditionDisc}) is fulfilled the obtained
set of probable words remains in the region of common beliefs, which
guarantees the validity of the subsequent approximation. This optimised
$\mathbf{v}$ is then used to update $\mathbf{w}^{(l)}$ under fixed
$\mathbf{w}^{(h)}$ and $h\neq l$. 

\begin{algorithm}[tbh]
\begin{enumerate}
\item Set $\mathbf{w}^{(1)}=\mathbf{w}^{(2)}=\mathbf{0}$. Set $l=2$ and
$h=1$.
\item Swap $l$ and $h$. Set $\mathbf{z}=\mathbf{w}^{(l)}$
\item Set $\mathbf{v}$ such that \[
\hat{H}(\mathbf{C}^{\otimes}|\mathbf{v}\Vert\mathbf{m})\to\min\]
 under $H(C_{i}|v_{i})\geq H(C_{i}^{(l)}|v_{i}\Vert\mathbf{m}^{(l)})$
for all $i$.
\item Set $\mathbf{w}^{(l)}=\mathbf{v}-\mathbf{w}^{(h)}-\mathbf{r}.$
\item If $\mathbf{w}^{(l)}\neq\mathbf{z}$ then go to Step 2.
\item Set $\hat{c}_{i}=\mbox{sign}(v_{i})$ for all $i$. 
\end{enumerate}
\caption{\label{alg:Iteration-with-Approximated}Iteration with Approximated
Discrimination}

\end{algorithm}
This gives Algorithm~\ref{alg:Iteration-with-Approximated}. Consider
first the constrained optimisation in Step~3 of Algorithm~\ref{alg:Iteration-with-Approximated}.
The definition of the  cross entropy \[
H(C_{i}^{(l)}|v_{i}\Vert\mathbf{m}^{(l)})=\log_{2}(2^{v_{i}}+2^{-v_{i}})-v_{i}\tanh_{2}(L_{i}^{(l)}(\mathbf{m}^{(l)}))\]
 transforms the constraint to $v_{i}\tanh_{2}(v_{i})\leq v_{i}\tanh_{2}(L_{i}^{(l)}(\mathbf{m}^{(l)})),$
which is equivalent to \[
|v_{i}|\leq|L_{i}^{(l)}(\mathbf{m}^{(l)})|\mbox{ and }\mbox{sign}(v_{i})=\mbox{sign}(L_{i}^{(l)}(\mathbf{m}^{(l)})).\]
 Moreover, the optimisation $\hat{H}(\mathbf{C}^{\otimes}|\mathbf{v}\Vert\mathbf{m})\to\min$
without constraint gives $v_{i}=\hat{L}_{i}^{\otimes}(\mathbf{m})$. 

\begin{samepage}This consideration directly gives the following cases:

\begin{itemize}
\item If this $v_{i}$ does not violate the constraint then it is already
 optimal. 
\item It violates the constraint if \[
\mbox{sign}(L_{i}^{(l)}(\mathbf{m}^{(l)}))\neq\mbox{sign}(\hat{L}_{i}^{\otimes}(\mathbf{m})).\]
In this case one has to set $v_{i}=0$ to fulfil the constraint. 
\item For the remaining case that $\mbox{sign}(L_{i}^{(l)}(\mathbf{m}^{(l)}))=\mbox{sign}(\hat{L}_{i}^{\otimes}(\mathbf{m}))$
but that the constraint is violated by \[
|L_{i}^{(l)}(\mathbf{m}^{(l)})|<|\hat{L}_{i}^{\otimes}(\mathbf{m})|\]
the optimal solution is \[
v_{i}=L_{i}^{(l)}(\mathbf{m}^{(l)})\]
 as the cross entropy $\hat{H}(C_{i}^{\otimes}|v_{i}\Vert\mathbf{m})$
is between $v_{i}=0$ and $v_{i}=\hat{L}_{i}^{\otimes}(\mathbf{m})$
a strictly monotonous function.
\end{itemize}
\end{samepage}

The obtained $v_{i}$ are thus given by either $\hat{L}_{i}^{\otimes}(\mathbf{m})$,
$L_{i}^{(l)}(\mathbf{m}^{(l)})$, or Zero. The zero value is hereby
obtained if the two estimated symbol decisions mutually contradict
each other, which is a rather intuitive result. Moreover, note that
the constrained optimisation is symmetric, i.e., it is equivalent
to \begin{equation}
H(C_{i}^{(l)}|v_{i}\Vert\mathbf{m}^{(l)})\to\min\mbox{ under }H(C_{i}|v_{i})\geq\hat{H}(C_{i}^{\otimes}\Vert\mathbf{m})\mbox{ for all }i.\label{eq:SymmetricOptimisation}\end{equation}

\begin{anm}
(\emph{Higher} \emph{Order} \emph{Moments}) By the central limit theorem
higher order moments do not significantly improve the approximation.
This statement is surprising as the knowledge of all moments leads
to perfect knowledge of the distribution and thus to globally maximal
discrimination. However, the statement just indicates that one would
need a large number of higher order moments to obtain additional useful
information about the distributions.
\end{anm}

\subsubsection{Convergence}

At the beginning of the algorithm many words are considered and a\noun{
Gauss} approximation surely suffices. I.e., in this case an approximation
by histograms would not produce significantly different results. The
convergence properties should thus at the beginning be comparable
to an algorithm that uses a discrimination via histograms. However,
there a sufficiently small $\varepsilon$ should give good convergence
properties.

At the end of the algorithm typically only few words remain to be
considered. For this case the \noun{Gauss} approximation is surely
outperformed by the use of histograms. Note, that this observation
does not contradict the statement of Lemma~\ref{lem:gauss-gauss}
as we there implicitly assumed {}``enough'' entropy. Intuitively,
however, this case is simpler to solve, which implies that the \noun{Gauss}
approximation should remain sufficient.

This becomes clear by reconsidering the region $\mathbb{A}_{n\varepsilon}(\mathbf{S}|\mathbf{m})$
that is employed in each step of the algorithm. An algorithm that
uses histograms will outperform an algorithm with a \noun{Gauss} approximation
if different independent regions in $\mathbb{A}_{n\varepsilon}(\mathbf{S}|\mathbf{m})$
become probable. A \noun{Gauss} approximation expects a connected
region and will thus span over these regions. I.e., the error of the
approximation will lead to a larger number of words that need to be
considered. However, this  should not have a significant impact on
the convergence properties. 

Typically the number of words to be considered will thus become smaller
in any step: The iterative algorithm gives that in every step the
discriminated cross entropy (see Definition~\ref{Def:(Discriminated-Cross-Entropy)})\[
\hat{H}(\mathbf{C}^{\otimes}|\mathbf{m}^{(n)}\Vert\mathbf{m}^{(o)})=\intop_{\mathbb{U}}\hat{p}^{\otimes}(\mathbf{u}|\mathbf{m}^{(o)})\log_{2}P(\mathbf{s}|\mathbf{m}^{(n)})\mathrm{d}\mathbf{u}\]
is smaller than the discriminated symbol entropy $\hat{H}(\mathbf{C}^{\otimes}\Vert\mathbf{m}^{(o)})$
under the assumed prior discrimination $\mathbf{m}^{(o)}$. Hence
the algorithm should converge to some fixed point.

\subsubsection{Fixed Points}

The considerations above give that the algorithm will typically not
stay in an infinite loop and thus end at a fixed point. Moreover,
at this fixed point the additional constraints will be fulfilled.
It remains to consider whether the additional constraints introduce
solutions of large discriminated symbol entropy $\hat{H}(\mathbf{C}^{\otimes}\Vert\mathbf{m})$. 

Intuitively the additionally imposed constraints seem not less restrictive
than the use of histograms and $\mathbf{w}^{(1)}=\mathbf{w}^{(2)}$
as $\mathbf{w}^{(1)}\neq\mathbf{w}^{(2)}$ implies a better discrimination.
However, solutions with large discriminated symbol entropy $\hat{H}(\mathbf{C}^{\otimes}\Vert\mathbf{m})$
will even for the second case typically not exist. Moreover, the discrimination
uses continuous values, which should be better than the again sufficient
hard decision discrimination. I.e., the constraint should have only
a small (negative) impact on the intermediate steps of the algorithm. 

Usually $\hat{H}(\mathbf{C}^{\otimes}\Vert\mathbf{m})$ is already
at the start ($\mathbf{w}^{(1)}=\mathbf{w}^{(2)}=\mathbf{0}$) relatively
small. The subsequent step will despite the constraint typically exhibit
a smaller discriminated symbol entropy. This is equivalent to a smaller
error probability and hence typically a better discrimination of the
distinguished word. 

If the process stalls for \[
\hat{H}(\mathbf{C}^{\otimes}\Vert\mathbf{m})>1\]
 then the by $\mathbf{m}$ investigated region either exhibits no
or multiple typical words. As (typically) the distinguished word is
the only code word in the typical set and as the typical set is (typically)
included this typically does not occur. 

Finally, for $\hat{H}(\mathbf{C}^{\otimes}\Vert\mathbf{m})\approx0$
the distinguished solution is found. At the end of the algorithm (and
the assumption of a distinguished solution) the obtained \noun{Gauss}
discriminated distribution then mimics a \noun{Gauss} approximation
of the overall distribution, i.e., \begin{equation}
\hat{p}^{\otimes}(\mathbf{u}|\mathbf{m})\propto\frac{\hat{p}^{(1)}(\mathbf{u}|\mathbf{m}^{(1)})\hat{p}^{(2)}(\mathbf{u}|\mathbf{m}^{(2)})}{\hat{p}(\mathbf{u}|\mathbf{m})}\asymp\hat{p}^{(a)}(\mathbf{u}|\mathbf{r}).\label{eq:Similard_iscriminated_gauss_distributions}\end{equation}
Without the constraints many solutions $\mathbf{m}$ exist. It only
has to be guaranteed that the constituent sets intersect at the distinguished
word. By the addition of the constraints the solution becomes unique
and is defined such that the number of by $\hat{p}(\mathbf{u}|\mathbf{m})$
considered words is as small as possible. 

This generally implies that then both constituent approximated distributions
need\emph{ }to be rather similar. This is the desired behaviour as
the considered environment is defined by a narrow peak of $\hat{p}^{(a)}(\mathbf{u}|\mathbf{r})$
around $\mathbf{u}(\hat{\mathbf{c}})$. Hence, the additional constraints
seem \emph{needed} for a defined fixed point and the limitations of
the \noun{Gauss} approximation. This emphasises the statement above:
Without the constraint non predictable behaviour may occur. 

\begin{anm}
(\emph{Optimality}) The values $w_{i}^{(l)}$ are continuous. Thus,
one can search for the optimum of \[
\hat{H}(\mathbf{C}^{\otimes}\Vert\mathbf{m})\to\min\]
 by a differentiation of $\hat{H}(\mathbf{C}^{\otimes}|\mathbf{m})$.
For the differentiation holds \[
2\frac{\partial\hat{H}(\mathbf{C}^{\otimes}\Vert\mathbf{m})}{\partial w_{i}^{(l)}}=\tanh_{2}L_{i}(\mathbf{m})-\tanh_{2}\hat{L}_{i}^{\otimes}(\mathbf{m})-2\intop_{\mathbb{U}}\frac{\partial\hat{p}^{\otimes}(\mathbf{u}|\mathbf{m})}{\partial w_{i}^{(l)}}\log_{2}P(\mathbf{s}|\mathbf{m})\mathrm{d}\mathbf{u}.\]
For the first term see Lemma~\ref{lem:Kullback} and the definition
of the discriminated symbol probabilities. The second term is the
derivation of the discriminated probability density. For the case
of a maximal discrimination of the distinguished word it will consist
of this word with probability of almost one. A differential variation
of the discriminator should remain maximally discriminating, which
gives that the second term should be almost zero. Hence, one obtains
that \begin{equation}
L_{i}(\mathbf{m})\approx\hat{L}_{i}^{\otimes}(\mathbf{m})\label{eq:gauss_well_defined}\end{equation}
 holds at the absolute minimum of $\hat{H}(\mathbf{C}^{\otimes}\Vert\mathbf{m})$,
which is a (soft decision) well defined discriminator. Note, furthermore,
that for (\ref{eq:Similard_iscriminated_gauss_distributions}) and
similar constituent distributions the distribution $\hat{p}(\mathbf{u}|\mathbf{m})$
will necessarily be similar to $\hat{p}^{(a)}(\mathbf{u}|\mathbf{m})$,
which is a similar statement as in (\ref{eq:gauss_well_defined}). 
\end{anm}

\begin{anm}
(\emph{Complexity}) The decoding complexity is under the assumption
of fast convergence of the order $O(n)$. I.e., the complexity \emph{only}
depends on the BCJR decoding complexity of the constituent codes.
Moreover, Algorithm~\ref{alg:Iteration-with-Approximated} can still
be considered as an algorithm where parameters are transferred between
the codes. Hereby the number of parameters is increased by a factor
of nineteen (for each $i$ additionally to $w_{i}^{(l)}$ for $x=\pm1$
three means and six correlations). 
\end{anm}
Note, finally, that the original iterative (constituent) belief propagation
algorithm is rather close to the proposed algorithm. Only by (\ref{eq:SymmetricOptimisation})
an additional constraint is introduced. Without the constraint apparently
too strong beliefs are transmitted. Algorithm~\ref{alg:Iteration-with-Approximated}
cuts off excess constituent code belief.

\subsection{Multiple Coupling}

Dually coupled codes constructed by just two constituent codes (with
simple trellises) are not necessarily good codes. This can be understood
by the necessity of simple constituent trellises. This gives that
the left-right (minimal row span)~\cite{DBLP:journals/tit/KschischangS95}
forms of the (permuted) parity check matrices have short effective
lengths. This gives that the codes cannot be considered as purely
random as this condition strongly limits the choice of codes. However,
to obtain asymptotically good codes one generally needs that the codes
can be considered as random.

If -- as in Remark~\ref{anm:(Multiple-Dual-Codes)} -- more constituent
codes are considered, then the dual codes will have smaller rate and
thus a larger effective length. This is best understood in the limit,
i.e., the case of $n-k$ constituent codes with $n-k^{(l)}=1$. These
codes can then be freely chosen without changing the complexity, which
leaves no restriction on the choice of the overall code. 

For a setup of a dual coupling with $N$ codes the discriminated distribution
of correlations is generalised to \[
P^{\otimes}(\mathbf{u}|\mathbf{m})\propto\frac{\prod_{l=1}^{N}P^{(l)}(\mathbf{u}|\mathbf{m}^{(l)})}{(P(\mathbf{u}|\mathbf{m}))^{N-1}},\]
with \[
\mathbf{m}=(\mathbf{r},\mathbf{w}^{(1)},\ldots,\mathbf{w}^{(N)})\mbox{, }\mathbf{m}^{(l)}=(\mathbf{r},\mathbf{w}^{(1)},\ldots,\mathbf{w}^{(l-1)},\mathbf{w}^{(l+1)},\ldots,\mathbf{w}^{(N)}),\]
and an independence assumption as in (\ref{eq:Independent}) and (\ref{eq:defwahrscheinlich}).

The definition of the discriminated symbol probabilities then becomes
\[
P_{C_{i}}^{\otimes}(x|\mathbf{m})\propto\sum_{\mathbf{u}}\frac{\prod_{l=1}^{N}P_{C_{i}}^{(l)}(x,\mathbf{u}|\mathbf{m}^{(l)})}{(P_{C_{i}}(x,\mathbf{u}|\mathbf{m}))^{N-1}}.\]
Moreover, for globally maximal discriminators \[
P_{C_{i}}^{\otimes}(x|\mathbf{m})=P_{C_{i}}^{(a)}(x|\mathbf{r})\mbox{ and }P^{\otimes}(\mathbf{u}|\mathbf{m})=P^{(a)}(\mathbf{u}|\mathbf{r})\]
remains true. The others lemmas and theorems above can be likewise
generalised. Hence, discriminator decoding by \noun{Gauss} approximations
applies to multiple dually coupled codes, too. 

\begin{anm}
(\emph{Iterative Algorithm}) The generalisation of Algorithm~\ref{alg:Iteration-with-Approximated}
may be done by using \begin{align*}
v_{i} & =\arg\min_{\mathbf{v}}H(C_{i}^{\otimes}|v_{i}\Vert\mathbf{m})\mbox{ under }H(C_{i}|v_{i})\geq H(C_{i}^{(l)}|v_{i}\Vert\mathbf{m}^{(l)})\mbox{ for all }i\\
\mathbf{w}^{(l)} & \leftarrow\mathbf{v}-\sum_{h\neq l}^{N}\mathbf{w}^{(l)}\end{align*}
as constituent code dependent update. 
\end{anm}
Overall this gives -- provided the distinguished well defined solution
is found -- that discriminator decoding asymptotically performs as
typical decoding for a random code. I.e., with dually coupled codes
and (to the distinguished solution convergent)\noun{ Gauss} approximated
discriminator decoding the capacity is attained.

\begin{anm}
(\emph{Complexity}) The complexity of decoding is of the order of
the sum of the constituent trellis complexities and thus generally
increases with the number of codes employed. For a fixed number of
constituent codes of fixed trellis state complexity and \noun{Gauss}
approximated discriminators the complexity thus remains of the order
$O(n)$. 
\end{anm}

\begin{anm}
(\emph{Number of Solutions}) For a coupling with many constituent
codes one obtains a large number of non linear optimisations that
have to be performed simultaneously. The non linearity of the common
problem should thus increase with the number of codes. Another explanation
is that then many times typicality is assumed. The probability of
some non typical event then increases. This may increase the number
of stable solutions of the algorithm or introduce instability. \\
This behaviour may be mitigated by the use of punctured codes. The
punctured positions define beliefs, too, which gives that the transfer
vector $\mathbf{w}$ is generally longer  than $n$. The transfer
complexity is thus increased, which should lead to better performance.
Note that this approach is implicitly used for LDPC codes. 
\end{anm}

%% file: TR-channel.tex
\subsection{Channel Maps}

In the last sections only memory-less channel maps as given in Remark~\ref{anm:BSC_AWGN}
were considered. A general channel is given by a stochastic map \[
\mathcal{K}:\mathbf{S}\to\mathbf{R}\mbox{ defined by }P_{\mathbf{R}|\mathbf{S}}(\mathbf{r}|\mathbf{s}).\]
We will here only consider channels where signal and {}``noise''
are independent. In particular we assume that the channel $\mathcal{K}$
is given by some known deterministic map \[
\mathcal{H}:\mathbf{s}\mapsto\mathbf{v}=(v_{1},\ldots,v_{n})\]
and $\mathbf{r}=\mathbf{v}+\mathbf{e}$ with the additive noise $\mathbf{E}$
defined by $P_{\mathbf{E}}(\mathbf{e})$. 

A code map $\mathcal{C}$ prior to the transmission together with
the map $\mathcal{H}$ may then be considered as a concatenated map.
The concatenation is hereby (for the formal representation by dually
coupled code see the proof of Theorem~\ref{thm:Both-direct-coupling-are-dual})
equally represented by the dual coupling of the {}``codes''\[
\mathbb{C}^{(1)}:=\{\mathbf{c}^{(1)}=(\mathbf{c},\mathbf{z}):\mathbf{c}\in\mathbb{C}\}\mbox{ and }\mathbb{C}^{(2)}:=\{\mathbf{c}^{(2)}=(\mathbf{s},\mathbf{v}):\mathbf{s}\in\mathbb{S}\mbox{ and }\mathcal{H}:\mathbf{s}\mapsto\mathbf{v}\}\]
where $\mathbf{z}=(z_{1},\ldots,z_{n})$ is undefined, i.e., no restriction
is imposed on $\mathbf{z}$. Moreover, $\mathbf{c}$ is punctured
prior to transmission and only $\mathbf{v}+\mathbf{e}$ is received.
Discriminator based decoding thus applies and one obtains \[
P_{C_{i}}^{\otimes}(x|\mathbf{m})\propto\sum_{\mathbf{u}\in\mathbb{U}}\frac{P_{C_{i}}^{(1)}(x,\mathbf{u}|\mathbf{w}^{(2)})P_{C_{i}}^{(2)}(x,\mathbf{u}|\mathbf{r},\mathbf{w}^{(1)})}{P_{C_{i}}(x,\mathbf{u}|\mathbf{w}^{(1)},\mathbf{w}^{(2)})}\]
as by the definition of the dually coupled code \[
P_{C_{i}}^{(1)}(x,\mathbf{u}|\mathbf{m}^{(1)})=P_{C_{i}}^{(1)}(x,\mathbf{u}|\mathbf{w}^{(2)})P(\mathbf{u}|\mathbf{r})\]
 and by the independence assumption $P_{C_{i}}(x,\mathbf{u}|\mathbf{m})=P_{C_{i}}(x,\mathbf{u}|\mathbf{w}^{(1)},\mathbf{w}^{(2)})P(\mathbf{u}|\mathbf{r})$
are independent of the channel. 

\begin{anm}
(\emph{Trellis}) If a trellis algorithm exists to compute $P_{C_{i}}^{(2)}(x|\mathbf{r})$
then one may compute the symbol probabilities $P_{C_{i}}^{(2)}(x|\mathbf{r},\mathbf{w}^{(1)})$,
the mean values and variances of $\mathbf{u}$ under $P_{C_{i}}^{(2)}(x,\mathbf{u}|\mathbf{r},\mathbf{w}^{(1)})$
with similar complexity.
\end{anm}
\begin{exa}
A linear time invariant channel with additive white \noun{Gauss}ian
noise $\underline{\mathbf{E}}(t)$ is given by the map\[
\underline{r}(t)=\int_{-\infty}^{\infty}\underline{s}(t-\tau)\underline{h}(\tau)\mathrm{d}\tau+\underline{e}(t).\]
Here, we assume a description in the equivalent base band. I.e., the
signals $\underline{\mathbf{r}}(t)$ and $\underline{\mathbf{s}}(t)$
as well as the noise may be complex valued -- indicated by the underbar.
The noise is assumed to be white and thus exhibits the (stationary)
correlation function $\mathrm{E}_{\underline{\mathbf{E}}(t)}[\underbar{e}(t)\underbar{e}^{\ast}(t+\tau)]=\sigma_{\underline{\mathbf{E}}(t)}^{2}\cdot\delta(\tau).$ 

For amplitude shift keying modulation one employs the signal\[
\underline{s}(t)=\sum_{i=-\infty}^{\infty}s_{i}w(t-i\mathrm{T})\mbox{ with }w(\tau)\mbox{ being the waveformer}.\]
 With a matched filter and well chosen whitening filter one obtains
an equivalent (generally complex valued) \emph{discrete} channel \begin{equation}
\mathcal{Q}:\mathbf{s}\mapsto\mathbf{r}\mbox{ with }r_{i}=\sum_{j=0}^{\mathrm{M}}s_{i-j}q_{j}+e_{i}\label{eq:Lin_channel_invariant}\end{equation}
defined by $\mathbf{q}=(q_{0},\ldots,q_{\mathrm{M}})$ and independent
\noun{Gauss} noise $\mathrm{E}_{\mathbf{E}}[e_{i}e_{j}^{*}]=\sigma_{E}^{2}\delta_{i-j}.$ 

For binary phase shift keying one has $s_{i}=\mathrm{A}x_{i}$ and
$x_{i}\in\mathbb{B}$.\\
 For quaternary phase shift keying the map is given by \[
s_{i}=\frac{A}{\sqrt{2}}(x_{2i}+\mathrm{j}x_{2i+1}),\]
$\mathrm{j}^{2}=-1$, and $x_{i}\in\mathbb{B}.$ In both cases a trellis
for $\mathbf{S}$ may be constructed with logarithmic complexity proportional
to the memory $\mathrm{M}$ of the channel $\mathbf{q}=(q_{0},q_{1},\ldots,q_{\mathrm{M}})$
times the number of information Bits per channel symbol $S_{i}$.
Note, moreover, that a time variance of the channel does not change
the trellis complexity.
\end{exa}

\subsection{Channel Detached Discrimination}

Overall one obtains for a linear modulation and linear channels with
additive noise the discrete probabilistic channel map \begin{equation}
\mathcal{K}:\mathbf{s}\mapsto\mathbf{r}=\mathbf{s}\mathbf{Q}+\mathbf{e}.\label{eq:linear_channel_general}\end{equation}
For uncorrelated \noun{Gauss} noise $\mathbf{E}$ this gives the probabilities
(without prior knowledge about the code words) %
{} \[
P(\mathbf{c}|\mathbf{r})\propto\exp_{2}(-\frac{\log_{2}(e)}{2\sigma_{E}^{2}}\Vert\mathbf{r}-\mathbf{c}\mathbf{Q}\Vert²).\]
 If the channel has large memory $\mathrm{M}$ and/or if a modulation
scheme with many Bits per symbol $s_{i}$ is used then the trellis
complexity of a trellis equalisation becomes prohibitively large.
To use the channel map as a constituent code will then not give a
practical algorithm.\\
Reconsider therefore the computation of the discriminated symbol probabilities
$P_{C_{i}}^{\otimes}(x|\mathbf{m})$ under the assumption that the
employed code is already a dually coupled code with, to again simplify
the notation, only two constituent codes. 

To apply the discriminator based approach one thus needs to compute
\[
P_{C_{i}}^{\otimes}(x|\mathbf{m})\propto\sum_{\mathbf{u}\in\mathbb{U}}\frac{P_{C_{i}}^{(1)}(x,\mathbf{u}|\mathbf{m}^{(1)})P_{C_{i}}^{(2)}(x,\mathbf{u}|\mathbf{m}^{(2)})}{P_{C_{i}}(x,\mathbf{u}|\mathbf{m})}.\]
Obviously, the complexity of the computation of the symbol probabilities
$P_{C_{i}}^{(l)}(x,\mathbf{u}|\mathbf{m}^{(l)})$ of the constituent
codes under the channel maps is generally prohibitively large. However,
one may equivalently (see (\ref{eq:axelsmeckerei}) on Page~\pageref{eq:axelsmeckerei})
compute \begin{equation}
P_{C_{i}}^{\otimes}(x|\mathbf{m})\propto\sum_{\mathbf{u}\in\mathbb{U}}\frac{P_{C_{i}}^{(1)}(x,\mathbf{u}|\mathbf{w}^{(2)})P_{C_{i}}^{(2)}(x,\mathbf{u}|\mathbf{w}^{(1)})}{P_{C_{i}}(x,\mathbf{u}|\mathbf{w}^{(1)},\mathbf{w}^{(2)})}\exp_{2}(u_{0})\label{eq:channelwithu_o}\end{equation}
where $u_{0}$ represents the channel probabilities. An alternative
method to compute $P_{C_{i}}^{\otimes}(x|\mathbf{m})$ is thus to
first compute \[
P_{C_{i}}^{\otimes}(x,\mathbf{u}|\mathbf{w}^{(1)},\mathbf{w}^{(2)})\propto\frac{P_{C_{i}}^{(1)}(x,\mathbf{u}|\mathbf{w}^{(2)})P_{C_{i}}^{(2)}(x,\mathbf{u}|\mathbf{w}^{(1)})}{P_{C_{i}}(x,\mathbf{u}|\mathbf{w}^{(1)},\mathbf{w}^{(2)})}\]
by the constituent distributions $P_{C_{i}}^{(l)}(x,\mathbf{u}|\mathbf{w}^{(h)})$
for $h\neq l$ and $P_{C_{i}}(x,\mathbf{u}|\mathbf{w}^{(1)},\mathbf{w}^{(2)})$
to then sum the by $\exp_{2}(u_{0})$ multiplied distributions $P_{C_{i}}^{\otimes}(x,\mathbf{u}|\mathbf{w}^{(1)},\mathbf{w}^{(2)}).$
In the distributions $P_{C_{i}}^{\otimes}(x,\mathbf{u}|\mathbf{w}^{(1)},\mathbf{w}^{(2)})$
the variable $u_{0}=\log_{2}(P(\mathbf{c}|\mathbf{r}))$ thereby relates
to the channel probabilities. The discrimination itself is detached
from the channel information, i.e., done only by the $\mathbf{w}^{(l)}$. 

This approach gives for linear channels and a \noun{Gauss} approximation
a surprisingly small complexity. This is the case as for linear channel
maps the computation of the {}``channel moments'', i.e., the moments
depending on $u_{0}$ is not considerably more difficult than the
computation of the code moments above. To illustrate consider the
channel dependent means, i.e., the expectation $\mathrm{E}_{C_{i}}^{(l)}[u_{0}|x,\mathbf{w}^{(h)}]$
where one obtains \begin{align}
\mathrm{E}_{C_{i}}^{(l)}[u_{0}|x,\mathbf{w}^{(h)}] & :=\sum_{\mathbf{c}\in\mathbb{C}_{i}^{(l)}(x)}u_{0}\, P(\mathbf{c}|\mathbf{w}^{(h)})=\sum_{\mathbf{c}\in\mathbb{C}_{i}^{(l)}(x)}\log_{2}(P(\mathbf{c}|\mathbf{r}))\, P(\mathbf{c}|\mathbf{w}^{(h)})\nonumber \\
 & \phantom{:}=\mathrm{E}_{C_{i}}^{(l)}[\log_{2}(P(\mathbf{c}|\mathbf{r}))|x,\mathbf{w}^{(h)}]=\mbox{const}+\frac{\log_{2}(e)}{2\sigma_{E}^{2}}\mathrm{E}_{C_{i}}^{(l)}[\Vert\mathbf{r}-\mathbf{c}\mathbf{Q}\Vert²|x,\mathbf{w}^{(h)}].\label{eq:estchannelmeans}\end{align}
This is similar to the computation of the variances on Page~\pageref{sub:Moments}.
Generally holds that the means and correlations can be computed for
linear channels with complexity that increases only linearly with
the channel memory $\mathrm{M}$. This result follows as the expectations
for the channels remain computations of moments, but now with vector
operations. The computation of the  variance of $u_{0}$ (for a channel
with memory) is, e.g., equivalent to the computation of a fourth moments
in the independent case. 

Generally holds that the mean values $\mathrm{E}_{C_{i}}^{(l)}[u_{0}|x,\mathbf{w}^{(h)}]$
are only computable up to a constant. This is under a \noun{Gauss}
assumption and (\ref{eq:channelwithu_o}) equivalent to a shift of
$u_{0}$ in $\exp_{2}(u_{0})$ by this constant. However, this will
lead to a proportional factor, which vanishes in the computation of
$P_{C_{i}}^{\otimes}(x|\mathbf{m})$. This unknown constant may thus
be disregarded. 

\begin{anm}
(\emph{Constituent} \emph{Code}) This approach applies by \[
P_{C_{i}}^{(l)}(x|\mathbf{m}^{(l)})\propto\sum_{\mathbf{u}\in\mathbb{U}}P_{C_{i}}^{(l)}(x,\mathbf{u}|\mathbf{w}^{(h)})\exp_{2}(u_{0})\mbox{ for }l\neq h\]
to the constituent codes $\mathbf{C}^{(l)}$, too: One may likewise
compute the constituent beliefs via the moments and a \noun{Gauss}
approximation and thus apply Algorithm~\ref{alg:Iteration-with-Approximated}.
\end{anm}
The \noun{Gauss} approximation for $u_{0}$ surely holds true if the
channel is short compared to the overall length as then many independent
parts contribute. With (\ref{eq:channelwithu_o}) one can thus apply
the iterative decoder based on \noun{Gauss} approximated discriminators
for linear channels with memory without much extra complexity.

\begin{anm}
(\emph{Matched Filter}) Note that one obtains by (\ref{eq:estchannelmeans})
for the initialisation $\mathbf{w}^{(l)}=\mathbf{0}$, $l=1,2$ that
$\hat{L}_{i}^{{\scriptscriptstyle \boxtimes}}(\mathbf{m})$ is proportional
to the {}``matched filter output'' given by $\mathbf{q}_{i}\mathbf{r}^{H}$.
Moreover, in all steps of the algorithm only $\hat{\mathbf{L}}^{{\scriptscriptstyle \boxtimes}}(\mathbf{m})$
is directly affected by the channel map. 
\end{anm}

\subsection{Estimation }

In many cases the transmission channel is unknown at the receiver.
This problem is usually mitigated by a channel estimation prior to
the decoding. However, an independent estimation needs -- especially
for time varying channels \cite{schnee} -- considerable excess redundancy.
The optimal approach would be to perform decoding, estimation, and
equalisation simultaneously. 

\begin{exa}
Assume that it is known that the channel is given as in (\ref{eq:Lin_channel_invariant}),
but that the channel parameters $\mathbf{q}=(q_{0},\ldots,q_{L})$
are unknown. Moreover, assume that the transmission is in the base
band, which gives that the $q_{i}$ are real valued. The aim is to
determine these values together with the code symbol decisions. To
consider them in the \emph{same} way, i.e., by decisions one needs
to reduce the (infinite) description entropy. We therefore assume
a quantisation of $\mathbf{q}$ by a binary vector $\mathbf{b}$.
This may, e.g., be done by\[
q_{i}=\mbox{q}\sum_{j=0}^{B_{i}-1}b_{l(i)+j}\exp_{2}(j),\mbox{ }l(i)=l(i-1)+B_{i-1}\mbox{, }l(0)=0\mbox{, and }b_{i}\in\mathbb{B}.\]
Note that one uses the additional knowledge $|q_{i}|<\mbox{q}\exp_{2}(B_{i})$
under this quantisation. Moreover, the quantisation error tends to
zero with the quantisation step size $\mbox{q}$. Finally, surely
a better quantisation can be found via rate distortion theory.
\end{exa}
The example shows that one obtains with an appropriate quantisation
additional binary unknowns $b_{j}$. Thus one needs additional parameters
$w_{n+j}^{(l)}$ that discriminate these Bits. Moreover, again a probability
distribution is needed for these $w_{n+j}^{(l)}$. Here it is assumed
that the distribution given in (\ref{eq:defwahrscheinlich}) is just
extended to these parameters. Note that this is equivalent to assuming
that code Bits $c_{i}$ and {}``channel Bits'' $b_{j}$ are independent. 

The code symbol discriminated probabilities remain under the now longer
$\mathbf{w}$ as in (\ref{eq:channelwithu_o}). Additionally one obtains
discriminated channel symbol  probabilities given by\[
P_{B_{i}}^{\otimes}(x|\mathbf{m})\propto\sum_{\mathbf{u}\in\mathbb{U}}\frac{P_{B_{i}}^{(1)}(x,\mathbf{u}|\mathbf{w}^{(2)})P_{B_{i}}^{(2)}(x,\mathbf{u}|\mathbf{w}^{(1)})}{P_{B_{i}}(x,\mathbf{u}|\mathbf{w}^{(1)},\mathbf{w}^{(2)})}\exp_{2}(u_{0}).\]
A \noun{Gauss} approximated discrimination is thus as before, however,
one needs to compute new and more general expectations. E.g., for
the general linear channel of (\ref{eq:linear_channel_general}) one
needs to compute the expectation given by\[
\mathrm{E}_{C_{i}}^{(l)}[u_{0}|x,\mathbf{w}^{(h)}]=\mbox{const}-\frac{\log_{2}(e)}{2\sigma_{E}^{2}}\mathrm{E}_{C_{i}}^{(l)}[\Vert\mathbf{r}-\mathbf{c}\mathbf{Q}(\mathbf{b})\Vert²|x,\mathbf{w}^{(h)}]\]
and equivalently for $\mathrm{E}_{B_{i}}^{(l)}[u_{0}|x,\mathbf{w}^{(h)}]$. 

The expectations are generalised because $\mathbf{Q}$ is a map of
the random variables $b_{j}$. With the quantisation of the example
above this map is linear in\textbf{ $\mathbf{b}$.} This first gives
that $\mathbf{c}\mathbf{Q}(\mathbf{b})$ can be considered as a quadratic
function in the binary random variables $\mathbf{x}$ and $\mathbf{b}$.
The computation of the means is thus akin to the one of fourth moments
and a known independent channel. 

Overall this gives  that the complexity for the computation of the
means and variances is for unknown channels {}``only'' twice as
large as for a known channel (of the same memory). It may, however,
still be computed with reasonable complexity. Hence, again an iteration
based on \noun{Gauss} approximated discriminators can be performed. 

\begin{anm}
(\emph{Miscellaneous}) Note that without some known {}``training''
sequence in the code word the iteration will by the symmetry usually
stay at $\mathbf{w}^{(l)}=0.$ Note, moreover, that this approach
is easily extended to time variant channels as considered in~\cite{schnee}
or even to more complex, i.e., non linear channel maps. The complexity
then remains dominated by the complexity of the computation of the
means and correlations.
\end{anm}

%% file: TR-summary.tex
\section{Summary}

In this paper first (dually) coupled codes were discussed. A dually
coupled code is given by a juxtaposition of the constituent parity
check matrices. Dually coupled codes provide a straightforward albeit
prohibitively complex computation of the overall word probabilities
$P^{(a)}(\mathbf{s}|\mathbf{r})$ by the constituent probabilities
$P^{(l)}(\mathbf{s}|\mathbf{r})$. However, for these codes a decoding
by belief propagation applies.

The then introduced concept of discriminators is summarised by augmenting
the probabilities by additional (virtual) parameters $\mathbf{w}^{(l)}$
and $\mathbf{u}$ to $P(\mathbf{s},\mathbf{u}|\mathbf{r},\mathbf{w}^{(1)},\mathbf{w}^{(2)})$.
This is similar to the procedure used for belief propagation but there
the parameter $\mathbf{u}$ is not considered. Such carefully chosen
probabilities led (in a globally maximal form) again to optimum decoding
decisions of the coupled code. However, the complexity of decoding
with globally maximal discriminators remains in the order of a brute
force computation of the ML decisions. 

It was then shown that local discriminators may perform almost optimally
but with  much smaller complexity. This observation then gave rise
to the definition of well defined discriminators and therewith again
an iteration rule. It was then shown that this iteration theoretically
admits any element of the typical set of the decoding problem as fixed
point. 

In the last chapter the central limit theorem then led to a \noun{Gauss}
approximation and a low complexity decoder. Finally (linear) channel
maps with memory were considered. It was shown that under additional
approximations equalisation and estimation may be accommodated into
the iterative algorithm with only little impact on the complexity.

%% file: TR-trellis.tex
\appendix

\section{Appendix}

\subsection{Trellis Based Algorithms\label{sub:Trellis-Based-Algorithms}}

\begin{floatingfigure}[r]{0.4\columnwidth}%
\noindent \begin{centering}
\includegraphics{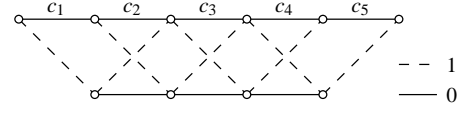}
\par\end{centering}

\caption{Trellis of the (5,4,2) Code}
\end{floatingfigure}%
The trellis is a layered graph representation of the code space $\mathbb{E}(\mathbf{C})$
such that every code word $\mathbf{c}=(c_{1},\ldots,c_{n})$ corresponds
to a unique path through the trellis from left to right. For a binary
code every layer of edges is labelled by one code symbol $c_{i}\in\mathbb{Z}_{2}=\{0,1\}$.
The complexity of the trellis is generally given by the maximum number
of edges per layer.\\
As example the trellis of a {}``single parity check'' code of length
$5$ with $\mathbf{H}=(11111)$ is depicted in the figure to the right.
Each of the $2^{4}$ paths in the trellis defines $c_{1}$ to $c_{5}$
of a code word $\mathbf{c}$ of even weight.

Here only the basic ideas needed to perform the computations in the
trellis are presented. A formal description will be given in another
paper~\cite{Axel}. The description here reflects the operations
performed in the trellis. I.e., only the \emph{lengthening} (extending
one path) and the \emph{junction} (combining two incoming paths of
one trellis node) are considered.

This is first explained for the \noun{Viterbi}~\cite{Forney:_Viterbi_1973}
algorithm that finds the code word with minimal distance. The {}``lengthening''
is given by an addition of the path correlations as depicted in Figure~\ref{fig:Basic-Operations}~(a).
For the combination -- the {}``join'' operation -- only the path
of maximum value is kept. This is equivalent to a minimisation operation
for the distances. This is reflected in the name of the algorithm,
which is often called \emph{min}-\emph{sum} algorithm. 

On the other hand, the BCJR~\cite{BCJR} algorithm (to compute $P_{C_{i}}^{(c)}(x|\mathbf{r})$)
is often called \emph{sum}-\emph{product} algorithm as the lengthening
is performed by the product of the path probabilities. The combination
of two paths is given by a sum. These operations are summarised in
Figure~\ref{fig:Basic-Operations} (b).%
\begin{figure}[H]
\begin{centering}
\subfigure[{\sc Viterbi} Algorithm]{\includegraphics{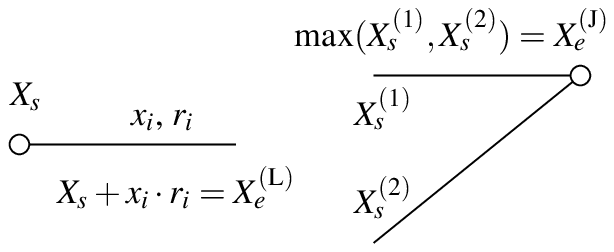}}~~~~~~~~~\subfigure[BCJR Algorithm]{\includegraphics{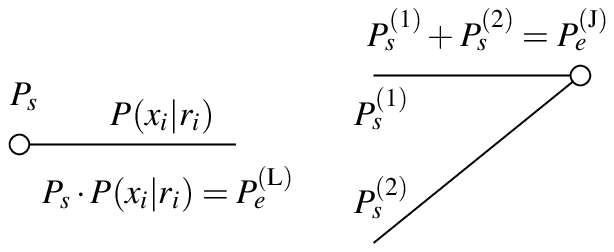}}
\par\end{centering}

\caption{\label{fig:Basic-Operations}Basic Operations in the \noun{Viterbi}
and BCJR Algorithms}

\end{figure}

\begin{anm}
(\emph{Forward}-\emph{Backward} \emph{Algorithm}) For the \noun{Viterbi}
algorithm the ML code word is found by following the selected paths
(starting from the end node) in backward direction. 

The operations of the BCJR algorithm (in forward direction) give at
the end directly the probabilities $P_{C_{n}}^{(c)}(x|\mathbf{r})$.
To compute all $P_{C_{i}}^{(c)}(x|\mathbf{r})$ the BCJR algorithm
has be performed into both directions. 

The same holds true for the algorithms below. This is here not considered
any further -- but keep in mind that only by this two way approach
symbol based distributions or moments can be computed with low complexity.
\end{anm}
In the following we shall reuse the notation of Figure~\ref{fig:Basic-Operations}
and use the indexes $s$ and $e$ before respectively after the lengthen
or join operation.

\subsubsection{Discrete Sets }

To compute a \emph{hard decision distribution} one can just count
the number of words of a certain distance to $\mathbf{r}$. Let this
number be denoted $D(t)$ for weight $t\in\mathbb{Z}$. 

This can be done in the trellis by using for the lengthening operation
from $D_{s}(t)$ to $D_{e}^{(\mathrm{L})}(t)$ by \begin{align*}
D_{e}^{(\mathrm{L})}(t) & =\begin{cases}
D_{s}(t-1) & \mbox{ for }c_{i}\neq r_{i}\\
D_{s}(t) & \mbox{ for }c_{i}=r_{i}.\end{cases}\end{align*}
The junction of paths becomes just \[
D_{e}^{(\mathrm{J})}(t)=D_{s}^{(1)}(t)+D_{s}^{(2)}(t).\]

Given $D(t)$ and a BSC with error probability $p$ one may compute
the probability of having words of distance $t$ by \[
P(t|\mathbf{r})\propto D(t)\cdot p^{t}(1-p)^{n-t}.\]
This can also be done directly in the trellis by\[
p_{e}^{(\mathrm{J})}(t)=p_{s}^{(1)}(t)+p_{s}^{(2)}(t)\mbox{ and }p_{e}^{(\mathrm{L})}(t)=\begin{cases}
p\cdot p_{s}(t-1) & \mbox{ for }c_{i}\neq r_{i}\\
(1-p)\cdot p_{s}(t) & \mbox{ for }c_{i}=r_{i}.\end{cases}\]

\subsubsection{Moments}

For the mean value \[
\mu=\mathrm{E}[\mathbf{r}\mathbf{c}^{T}]=\sum_{i=1}^{n}\mathrm{E}[r_{i}c_{i}]\mbox{ holds }\mathrm{E}[\sum_{j=1}^{i}c_{j}r_{j}|c_{i}]=\mathrm{E}[\sum_{j=1}^{i-1}c_{j}r_{j}]+c_{i}r_{i}.\]
This directly gives that one obtains for the lengthening\[
P_{e}^{(\mathrm{L})}=P_{s}\cdot2^{rc_{i}}\mbox{ and }\mu_{e}^{(\mathrm{L})}=\mu_{s}+r_{i}c_{i}.\]
The junction is just the probability weighted sum of the prior computed
input means given by \[
P_{e}^{(\mathrm{J})}=P_{s}^{(1)}+P_{s}^{(2)}\mbox{ and }\mu_{e}^{(\mathrm{J})}=\frac{P_{s}^{(1)}}{P_{e}^{(L)}}\mu_{s}^{(1)}+\frac{P_{s}^{(2)}}{P_{e}^{(L)}}\mu_{s}^{(2)}.\]
Hence, the BCJR algorithm for the probabilities needs to be computed
at the same time. Note that the obtained mean values are then readily
normalised. 

To compute the {}``energies'' $\mathrm{S}=\mathrm{E}[(\sum_{j=1}^{i}c_{j}r_{j})^{2}]$
one uses in the same way that \[
\mathrm{S}=\mathrm{E}[(\sum_{j=1}^{i}c_{j}r_{j})^{2}|c_{i}]=\mathrm{E}[(\sum_{j=1}^{i-1}c_{j}r_{j})^{2}]+2c_{i}r_{i}\cdot\mathrm{E}[\sum_{j=1}^{i-1}c_{j}r_{j}]+(c_{i}r_{i})^{2}.\]
This additionally gives -- to the then necessary computation of means
and probabilities -- that lengthening and junction are now given by\[
S_{e}^{(\mathrm{L})}=S_{s}+2r_{i}c_{i}\cdot\mu_{s}+(r_{i}c_{i})^{2}\mbox{ and }S_{e}^{(\mathrm{J})}=\frac{P_{s}^{(1)}}{P_{e}^{(L)}}S_{s}^{(1)}+\frac{P_{s}^{(2)}}{P_{e}^{(L)}}S_{s}^{(2)}.\]
Here again the normalisation is already included. Correlation and
higher order moment trellis computations are derived in the same way.
However, for an $l-$th moment all $l-1$ lower moments and the probability
need to be additionally computed. Moreover, the description gives
that these moment computations may be performed likewise for any linear
operation $\mathbf{c}\mathbf{Q}$ (defined over the field of real
or complex numbers) then using vector operations.

\subsubsection{Continuous Sets}

Another possibility to use the trellis is to compute (approximated)
histograms for $u=\mathbf{w}\mathbf{c}^{T}$ with $w_{i}\in\mathbb{R}$
and $c_{i}\in\mathbb{B}$. It is here proposed (other possibilities
surely exist) to use -- as in the hard decision case above -- a vector
function $(h(t),\mu)$ with $t\in\mathbb{Z}$ and $|t|\leq\mathrm{Q}$
and the mean value $\mu$. I.e., the values of $u$ with non vanishing
probability are assumed to be in a vicinity the mean value $\mu$
(computed above) or\[
p(u|\mathbf{m})=0\mbox{ for }|u-\mu|>\mathrm{Q}\varepsilon.\]
Thus $(h(t),\mu)$ is defined to be the approximation of \[
h(t)\approx\int_{t\varepsilon}^{(t+1)\varepsilon}\!\!\! p(u-\mu|\mathbf{m})\mathrm{d}u.\]
Here, densities are used to simplify the notation. It is now assumed
that the mean values are computed as above, which gives that the lengthening
is the trivial operation \[
(h_{e}^{(\mathrm{L})}(t),\mu^{(\mathrm{L})})=(h_{s}(t),\mu_{s}+c_{i}w_{i}).\]
The junction, however, cannot be easily performed as usually the mean
values do not fit on each other. Here, it is assumed that the density
has for any interval the form of a rectangle. Note that this is again
a maximum entropy assumption.

This gives the approximation of the histogram $h_{e}^{(\mathrm{J})}(t)$
by the junction operation to be \[
(h_{e}^{(J)}(t),\mu^{(\mathrm{J})})=(\frac{P_{s}^{(1)}}{P_{e}^{(L)}}\breve{h}_{s}^{(1)}(t)+\frac{P_{s}^{(2)}}{P_{e}^{(L)}}\breve{h}_{s}^{(2)}(t),\frac{P_{s}^{(1)}}{P_{e}^{(L)}}\mu_{s}^{(1)}(t)+\frac{P_{s}^{(2)}}{P_{e}^{(L)}}\mu_{s}^{(2)}).\]
 and $ $\[
\breve{h}_{s}^{(j)}(t-\left\lfloor (\mu_{s}^{(j)}-\mu_{e}^{(\mathrm{L})})\varepsilon\right\rfloor )=a(\mu_{s}^{(j)},\mu_{e}^{(\mathrm{L})})\cdot h_{s}^{(j)}(t)+b(\mu_{s}^{(j)},\mu_{e}^{(\mathrm{L})})\cdot h_{s}^{(j)}(t+1),\]
with $\left\lfloor z\right\rfloor $ the integer part, $\mbox{trunc}(z):=z-\left\lfloor z\right\rfloor $,
\[
a(\mu_{s}^{(j)},\mu_{e}^{(\mathrm{L})})+b(\mu_{s}^{(j)},\mu_{e}^{(\mathrm{L})})=1,\mbox{ and }b(\mu_{s}^{(j)},\mu_{e}^{(\mathrm{L})})=\mbox{trunc(}(\mu_{s}^{(j)}-\mu_{e}^{(\mathrm{L})})\varepsilon).\]

\subsection[\noun{Gauss}ian Approximated Belief]{Computation of $\hat{L}_{i}^{\otimes}(\mathbf{m})$\label{sub:Appendix_a2}}

Equation (\ref{eq:Pboxtimes}) gives the logarithmic probability ratio
\[
\hat{L}_{i}^{\otimes}(\mathbf{m})=r_{i}+\breve{L}_{i}^{(1)}(\mathbf{m}^{(1)})+\breve{L}_{i}^{(2)}(\mathbf{m}^{(1)})+\hat{L}_{i}^{{\scriptscriptstyle \boxtimes}}(\mathbf{m}).\]

The first three terms can be computed as before. For the computation
of $\hat{L}_{i}^{{\scriptscriptstyle \boxtimes}}(\mathbf{m})$ use
that \begin{equation}
\hat{P}_{C_{i}}^{{\scriptscriptstyle \boxtimes}}(x|\mathbf{m})\propto\intop_{\mathbb{U}}\frac{\hat{p}_{C_{i}}^{(1)}(\mathbf{u}|x,\mathbf{m}^{(1)})\cdot\hat{p}_{C_{i}}^{(2)}(\mathbf{u}|x,\mathbf{m}^{(2)})}{\hat{p}_{C_{i}}(\mathbf{u}|x,\mathbf{m})}\mathrm{d}\mathbf{u}=:\intop_{\mathbb{U}}\hat{p}_{C_{i}}^{{\scriptscriptstyle \boxtimes}}(x,\mathbf{u}|\mathbf{m})\mathrm{d}\mathbf{u}.\label{eq:GaussDiscProps}\end{equation}
To compute (\ref{eq:GaussDiscProps}) a multiplication of multivariate
Gauss distributions has to be performed.  The moments of the multivariate
distributions $\hat{p}_{C_{i}}^{(l)}(\mathbf{u}|x,\mathbf{m}^{(l)})$
and $\hat{p}_{C_{i}}(\mathbf{u}|x,\mathbf{m})$ are defined by \[
\mu_{i,j}^{(l)}(x)=\mathrm{E}_{\mathbf{C}^{(l)}|C_{i}}[u_{j}|x,\mathbf{m}^{(l)}]\mbox{ and }A_{i,j,k}^{(l)}(x)=\mathrm{E}_{\mathbf{C}^{(l)}|C_{i}}[(u_{j}-\mu_{i,j}^{(l)})(u_{k}-\mu_{i,k}^{(l)})|x,\mathbf{m}^{(l)}]\]
and likewise for $\mu_{i,j}(x)$ and $A_{i,j,k}(x)$. 

The  multivariate \noun{Gauss} distributions are of the form \[
\hat{p}_{C_{i}}(\mathbf{u}|x,\mathbf{m})=\frac{1}{\sqrt{|2\pi\mathbf{A}_{i}(x)|}}\exp\left(-(\mathbf{u}\!-\!\mathbf{\mu}_{i}(x))[2\mathbf{A}_{i}(x)]^{-1}\!(\mathbf{u}\!-\!\mathbf{\mu}_{i}(x))^{T}\right).\]
 Set \[
\mathbf{B}_{i}^{(l)}\!(x)=\left[\mathbf{A}_{i}^{(l)}\!(x)\right]^{-1}\!\mbox{ and }\mathbf{B}_{i}(x)=\left[\mathbf{A}_{i}(x)\right]^{-1}.\]
The operation in (\ref{eq:GaussDiscProps}) then leads to \[
\hat{p}_{C_{i}}^{\,{\scriptscriptstyle \boxtimes}}(x,\mathbf{u}|\mathbf{m})=\frac{\exp\left(\hat{\mbox{C}}_{i}^{\,{\scriptscriptstyle \boxtimes}}\!(x,\mathbf{m})-(\mathbf{u}\!-\,\hat{\mathbf{\mu}}_{i}^{\,{\scriptscriptstyle \boxtimes}}\!(x))\left[2\hat{\mathbf{A}}_{i}^{{\scriptscriptstyle \boxtimes}}\!(x)\right]^{-1}\!(\mathbf{u}\!-\!\hat{\mathbf{\mu}}_{i}^{\,{\scriptscriptstyle \boxtimes}}{\!(x))}^{T}\right)}{\sqrt{|2\pi\hat{\mathbf{A}}_{i}^{{\scriptscriptstyle \boxtimes}}\!(x)|}}\]
 with \[
\left[\hat{\mathbf{A}}_{i}^{{\scriptscriptstyle \boxtimes}}\!(x)\right]^{-1}=\mathbf{B}_{i}^{(1)}\!(x)+\mathbf{B}_{i}^{(2)}\!(x)-\mathbf{B}_{i}(x)\]
by a comparison of the terms $\mathbf{u}(.)\mathbf{u}^{T}$,\[
\hat{\mathbf{\mu}}_{i}^{{\scriptscriptstyle \boxtimes}}\!(x)=\left(\mathbf{\mu}_{i}^{(1)}\!(x)\mathbf{B}_{i}^{(1)}\!(x)+\mathbf{\mu}_{i}^{(2)}\!(x)\mathbf{B}_{i}^{(2)}\!(x)-\mathbf{\mu}_{i}(x)\mathbf{B}_{i}(x)\right)\hat{\mathbf{A}}_{i}^{{\scriptscriptstyle \boxtimes}}\!(x),\]
by a comparison of the in $\mathbf{u}$ linear terms, and\[
\begin{array}{ccc}
2\hat{\mbox{C}}_{i}^{{\scriptscriptstyle \boxtimes}}\!(x,\mathbf{m}) & = & \hat{\mathbf{\mu}}_{i}^{{\scriptscriptstyle \boxtimes}}\!(x)\left[\hat{\mathbf{A}}_{i}^{{\scriptscriptstyle \boxtimes}}(x)\right]^{-1}\!\hat{\mathbf{\mu}}_{i}^{{\scriptscriptstyle \boxtimes}T}\!(x)+\mathbf{\mu}_{i}(x)\mathbf{B}_{i}(x)\mathbf{\mu}_{i}^{T}(x)-\mathbf{\mu}_{i}^{(1)}\!(x)\mathbf{B}_{i}^{(1)}\!(x)\mathbf{\mu}_{i}^{(1)T}\!(x)\\
 &  & -\mathbf{\mu}_{i}^{(2)}\!(x)\mathbf{B}_{i}^{(2)}\!(x)\mathbf{\mu}_{i}^{(2)T}\!(x)-\log{\displaystyle \frac{|\mathbf{A}_{i}^{(1)}\!(x)||\mathbf{A}_{i}^{(2)}\!(x)|}{|\hat{\mathbf{A}}_{i}^{{\scriptscriptstyle \boxtimes}}\!(x)||\mathbf{A}_{i}(x)|}}\end{array}\]
by a consideration of the remaining constant.

From the definition of the multivariate distributions then follows
that \[
\hat{P}_{C_{i}}^{{\scriptscriptstyle \boxtimes}}(x|\mathbf{m})\propto\intop_{\mathbb{U}}\hat{p}_{C_{i}}^{{\scriptscriptstyle \boxtimes}}(x,\mathbf{u}|\mathbf{m})\mathrm{d}\mathbf{u}=\exp(\hat{\mbox{C}}_{i}^{{\scriptscriptstyle \boxtimes}}\!(x|\mathbf{m})),\]
respectively $\hat{L}_{i}^{{\scriptscriptstyle \boxtimes}}(\mathbf{m})={\displaystyle \frac{1}{2}}\mbox{log}_{2}(e)\cdot(\hat{\mbox{C}}_{i}^{{\scriptscriptstyle \boxtimes}}\!(+1|\mathbf{m})-\hat{\mbox{C}}_{i}^{{\scriptscriptstyle \boxtimes}}\!(-1|\mathbf{m})).$